\documentclass{llncs}

\usepackage{amsmath,amsfonts,amssymb}
\usepackage{url}

\usepackage{xcolor}
\usepackage{graphicx}

\newcommand\R{{\mathbb{R}}}

\newcommand{\set}[1]{\left\{{#1}\right\}}

\newcommand{\remove}[1]{}

\newcommand{\wt}{\widetilde}

\renewcommand{\epsilon}{\varepsilon}

\renewcommand{\leq}{\leqslant}
\renewcommand{\geq}{\geqslant}

\newcommand{\half}{\frac{1}{2}}

\newcommand{\uone}{{(1)}}
\newcommand{\utwo}{{(2)}}
\newcommand{\ui}{{(i)}}
\newcommand{\uj}{{(j)}}
\newcommand{\un}{{(n)}}

\newcommand{\littleo}[1]{\ensuremath{\operatorname{o}\bigl(#1\bigr)}}

\newsavebox{\savepar}
\newenvironment{boxit}[0]
 { \begin{center}\begin{lrbox}{\savepar} \begin{minipage}[t]{4.8in}}
 { \end{minipage}\end{lrbox}\fbox{\usebox{\savepar}} \end{center}}
\newcommand{\bcaption}[1]{\begin{center}
                         {\underline{\textit{#1}}}
                          \end{center}}

\begin{document}

\title{Stochastic Service Placement}

\author{Galia Shabtai
\and Danny Raz
\and Yuval Shavitt
}
\institute{Tel-Aviv University and Technion \\\email{galiashabtai@gmail.com},\,\, \email{danny@cs.technion.ac.il},\,\, \email{shavitt@eng.tau.ac.il}}

%\author{Galia Shabtai
%\thanks{
%xx.} \and Danny Raz
%\thanks{
%xxx}
%\and Yuval Shavitt
%\thanks{
%xxx}
%}
%\institute{The Faculty of Engineering,\\ Tel-Aviv University, Israel 69978\\\email{galiashabtai@gmail.com},\,\, \email{danny@cs.technion.ac.il},\,\, \email{shavitt@eng.tau.ac.il}}

\maketitle

\begin{abstract}
Resource allocation for cloud services is a complex
task due to the diversity of the services and the dynamic
workloads. One way to address this is by overprovisioning which
results in high cost due to the unutilized resources. A much
more economical approach, relying on the stochastic nature of
the demand, is to allocate just the right amount of resources and
use additional more expensive mechanisms in case of overflow
situations where demand exceeds the capacity. In this paper we
study this approach  and show both by comprehensive analysis for
independent normal distributed demands  and simulation on synthetic data that it is significantly better than
currently deployed methods.
\end{abstract}

%\clearpage
%\end{titlepage}
%
%\pagebreak

%% A category with the (minimum) three required fields
%\category{C.2}{COMPUTER-COMMUNICATION NETWORKS}{Miscellaneous}
%%A category including the fourth, optional field follows...
%%\category{D.2.8}{Software Engineering}{Metrics}[complexity measures, performance measures]
%
%\terms{Performance, Algorithms}
%
%\keywords{cloud, resource allocation} % NOT required for Proceedings

\section{Introduction} \label{sec:intro}

The recent rapid development of cloud technology gives rise to ``many-and-diverse'' services being deployed in datacenters across the world. The allocation of the available resources in the various locations to these services has a critical impact on the ability to provide a ubiquitous cost-effective high quality service.
There are many challenges associated with optimal service placement due to the large scale of the problem, the need to
obtain state information, and the geographical spreading of the datacenters and users.

One intriguing problem is the fact that the service resource requirement changes over time and is not fully known at the time of placement.
A popular way of addressing this important problem is over-provisioning, that is allocating resources for the peak demand. Clearly, this is not a cost effective approach as much of the resources are unused most of the time. An alternative approach is to model the service requirement as a stochastic process with known parameters (these parameters can be inferred from historical data).

Much of the previous work that used stochastic demand modeling  \cite{KRT00,GI99,WMZ11,BE12} attempted to minimize the probability of an overflow event and not the cost of this overflow. Their solution was based on solving the Stochastic Bin Packing (SBP) problem, where the goal is to pack a given set of random items
to a minimum number of bins such that each bin overflow probability will not exceed a given value.\footnote{\cite{WMZ11,BE12} also looked at the online SBP problem, where the items (services in our case) are assigned to some bin (datacenter) as they arrive.}

Thus, in these works there is no distinction between cases with marginal and substantial overflow of the demand, and under this modeling one again prepares for the worst and tries to prevent an overflow. However, in reality it is not cost effective to allocate resources according to the worst case scenario. Instead, one often wishes to minimize the cost associated with periods of insufficient resource availability,
which is commonly dealt with by either diverting the service request to a remote location or dynamically buying additional resources.
In both cases, the cost associated with such events is proportional to the amount of unavailable resources.
Thus, our aim is to minimize the expected overflow of the demand over time.

%and that the distribution, or at least its important parameters, can be learned....,
%in particular one can have good estimation of the average and the variance of the demand.
%We note, that in some cases the demand is deterministic and then the variance is zero.

We deviate from previous work in two ways. First, We look at a stochastic packing problem
where the number of bins is given, e.g., a company already has two datacenters where services can be placed, and one would
like to place services in these two datacenters optimally. Second, as we mentioned before, we look at a more practical optimality criterion
where we do not wish to optimize the probability of an overflow, but instead its expected deviation.

To better understand this,
assume that the resource we are optimizing is the bandwidth consumption of the service in a datacenter where we have a prebooked bandwidth
for each datacenter. Since we are dealing with stochastic bandwidth demand, with some probability the traffic will exceed the
prebooked bandwidth and will result with expensive overpayment that depends on the amount of oversubscribed traffic.
Obviously, in such a case minimizing the probability of oversubscription may not give the optimal cost,
and what we need to minimize is the expected deviation from the prebooked bandwidth, thus we term our problem
SP-MED (stochastic packing with minimum expected deviation).

We analyze the case of independent normal distribution and develop algorithms for the optimal partition between two or more datacenters that need not be identical. We prove the correctness of our algorithm by separating the problem into two: one dealing with the continuous characteristics of the stochastic objective function and the other dealing with the discrete nature of the combinatorial algorithm. This approach results in a clean and elegant algorithm and proof.

In fact, our proof technique reveals that the algorithms we developed hold for a large family of optimization criteria, and thus may be applicable for other optimization problems. In particular we study two other natural cost functions that correspond to minimizing overflow probability (rather than expected overflow deviation). We show that these cost functions also fall into our general framework, and therefore the same algorithm works for them as well! As will become clear later the requirements from the cost functions we have are quite natural, and we believe our methods will turn useful for many other applications.

\section{Related work}
\label{sec:related_work}

Early work on VM placement (e.g., \cite{AT07,BKB07,MN08}) models the problem as a deterministic bin packing problem, namely, for every service there is an estimate of its \emph{deterministic} demand. Stochastic bin packing was first suggested by Kleinberg, Rabani and Tardos \cite{KRT00} for statistical multiplexing. \cite{KRT00} mostly considered Bernoulli-type distributions. Goel and Indyk \cite{GI99} further studied Poisson and exponential distributions. Wang, Meng and Zhang \cite{WMZ11} suggested to model real data with the \emph{normal distribution}. Thus, the input to their stochastic packing problem
%(either SB-MOP, SP-MED, or any other stochastic problem covered by this analysis)
%is a set of $k$ bins with capacities $c_1,\ldots,c_k$ and
is $n$ independent services each with demand requirement distributed according to a distribution $X^\ui$ that is normal with mean $\mu^\ui$ and variance $V^\ui$. The output is some partition of the services to bins in a way that minimizes a target function that differs from problem to problem.
% (See, e.g., Section \ref{sec:SPMED:problem} for the definition of SP-MED and Section \ref{app:SPMOP} for SP-MOP).

A naive approach to such a problem is to reduce it to classical bin packing as follows: for the $i$'th service define the \emph{effective size} as the number $e^\ui$ such that the probability that $X^\ui$ is larger than $e^\ui$ is small; then solve the classical bin packing problem (or a variant of it) with item sizes $e^\uone,\ldots,e^\un$. However, \cite{WMZ11} showed this approach can be quite wasteful, mostly because it adds extra space per service and not per bin. To demonstrate the issue, think about unbiased, independent coin tosses. The probability one coin toss significantly deviates from its mean is $1$, while the probability $100$ independent coin tosses significantly deviate from the mean is exponentially small. When running independent trials there is a \emph{smoothing} effect that considerably reduces the chance of high deviations. This can also be seen from the fact that the standard deviation of $n$ independent, identical processes is only $\sqrt{n}$ times the standard deviation of one process, and so the standard deviation grows much slower than the number of processes.

Breitgand and Epstein \cite{BE12}, building on \cite{WMZ11}, suggest an algorithm for stochastic bin packing that takes advantage of this smoothing effect. The algorithm assumes all bins have equal capacity. The algorithm first sorts the processes by their variance to mean ratio (VMR), i.e., $\frac{V^\uone}{\mu^\uone} \le  \frac{V^\utwo}{\mu^\utwo} \le \ldots \le
\frac{V^\un}{\mu^\un}$. Then the algorithm finds the largest prefix
of the sorted list such that allocating that set of services to the
first bin makes the probability the first bin overflows at most
$p$. The algorithm then proceeds bin by bin, each time
allocating a prefix of the remaining services on the sorted list to the next bin. \cite{BE12} show that if
we allow \emph{fractional} solutions, i.e., we allow splitting services between bins, the algorithm finds an
optimal solution, and also show an online, integral version that
gives a $2$-approximation to the optimum.

\section{The risk unbalancing principle: A bird's overview of our technique}
\label{sec:risk_unbalancing_principle}

In this bird's overview we focus on the SP-MED problem, and in the next section we generalize it to a much wider class of cost functions.
As mentioned above we are given $k$ bins with capacities $c_1,\ldots,c_k$.
We are looking for a solution that minimizes the expected deviation. As before, we study the case where the stochastic demands are independent and normally distributed. We develop a new general framework to analyze this problem that also sheds light on previous work.

We first observe that an optimal fractional solution to SP-MED is also optimal for any two bins. This is true because the cost function is the sum of the expected deviation of all the bins, and changing the internal allocation of two bins only affects them and not the other bins. Thus, one possible approach to achieve an optimal solution, is by repeatedly improving the internal division of two bins. However, offhand, there is no reason to believe such a sequence of local improvements efficiently converges to the optimal solution. Surprisingly, this is indeed the case. Therefore, we first focus on the two bin case. Later (in Appendix \ref{sec:k}) we will see that
solving the $k=2$ case implies a solution to the general case of arbitrary $k$.

We recall that if we have $n$ independent normally distributed services with mean and variance
$(\mu^\ui,V^\ui)$, and we allocate the services with indices in $I
\subseteq [n]$ to the first bin, and the rest to the second bin,
then the first bin is normally distributed with mean
$\mu_1=\sum_{i \in I} \mu^\ui$ and variance $V_1=\sum_{i \in
I} V^\ui$, while the second one is normally distributed with mean $\mu-\mu_1$ and variance $V-V_1$ where $\mu=\sum_{i=1}^{n} \mu^
\ui$ and $V=\sum_{i=1}^{n} V^\ui$. We wish to minimize the total expected deviation of the two bins.

Let us define the function $Dev:[0,1] \times [0,1] \to \R$ such that $Dev(a,b)$ is the total expected deviation when bin one is distributed with mean $a\mu$ and variance $bV$ and bin two with mean $(1-a)\mu$ and variance $(1-b)V$.  Figure \ref{fig:normal_dev_a_b_c200_m160_v6400} (left) depicts $Dev(a,b)$ for two bins with equal capacity.
 $Dev$ has a reflection symmetry around the $a=b$ line, that corresponds to the fact that we can change the order of
the bins as long as they have equal capacity and remain with a
solution of the same cost (see Appendix \ref{app:spmed:symmetry} for
a proof of a similar symmetry for the case of different
capacity bins). The points $(0,0)$ and $(1,1)$ correspond to
allocating all the services to a single bin, and indeed the expected
deviation there is maximal.

The point $(\frac{1}{2},\frac{1}{2})$ is a \emph{saddle point} (see
Appendix \ref{app:spmed:saddle} for a proof). Figure
\ref{fig:normal_dev_a_b_c200_m160_v6400} (middle) shows a zoom in around
this saddle point. There are big mountains in the lower left and
upper right quarters that fall down to the saddle point, and also
valleys going down from the saddle point to the bottom right and top
left. Going back to Figure \ref{fig:normal_dev_a_b_c200_m160_v6400} (left)
we see that for every fixed $b$ the function is convex in $a$, with
a single minimum point, and all these minimum points form the two
valleys mentioned above.

\begin{figure}
  \begin{center}
  \includegraphics[width=4cm]{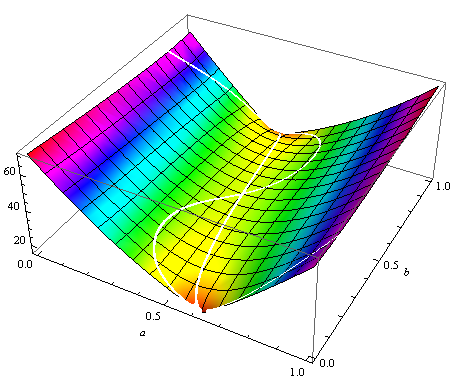}
  \includegraphics[width=4cm]{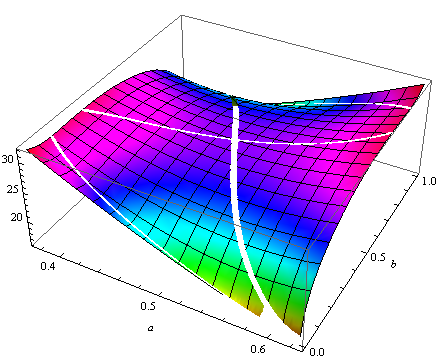}
  \includegraphics[width=3cm]{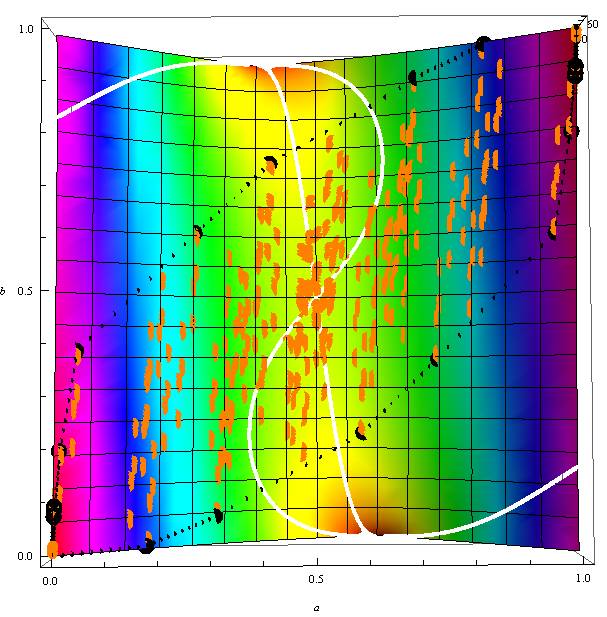}
  \end{center}
  %\caption{The left figure depicts $Dev(a,b)$ when $\mu=160$, $V=6400$ and $c_1=c_2=100$. The right figure is a zoom in around the saddle point $(\half,\half)$. The white lines represent the zero level sets of the partial derivatives.}
  \caption{The left figure depicts $Dev(a,b)$ when $\mu=160$, $V=6400$ and $c_1=c_2=100$. The middle figure is a zoom in around the saddle point $(\half,\half)$. The white lines represent the zero level sets of the partial derivatives. The right figure depicts $Dev(a,b)$ with the $2^{10}$ possible partitions represented in orange. The bottom sorted path (that sorts services by their VMR in \emph{increasing} order) and the upper sorted path (that sorts services by their VMR in \emph{decreasing} order), together with their integral points, are in black. Notice that all partition points are confined by both sorted paths.}
  \label{fig:normal_dev_a_b_c200_m160_v6400}
\end{figure}

We now go back and consider the input
$\set{(\mu^\ui,V^\ui)}_{i=1}^n$. We can represent service
$i$ with the pair $(a^\ui=\frac{\mu^\ui}{\mu},b^\ui=\frac{V^\ui}{V})$. If we split the input to the sets
$I$ and $[n] \setminus I$, the expected deviation of this partition is
$Dev(\sum_{i \in I} a^\ui,\sum_{i \in I} b^\ui)$. Thus, the $n$ input
points induce $2^n$ possible solutions $P_I=(\sum_{i \in I} a^\ui,\sum_{i
\in I} b^\ui)$ for each $I \subseteq [n]$, and our task is to find the partition
$I$ that minimizes $Dev$. Sorting the services by their VMR, is equivalent to sorting the vectors $P^\ui=(a^\ui,b^\ui)$ by their
angle with the $a$ axis. If the sorted list is
$P^\uone,\ldots,P^\un$ with increasing angles, looking for a solution in a partition that breaks the sorted sequence to two consecutive parts finds a solution among
$(0,0),P^\uone,P^\uone+P^\utwo,\ldots,P^\uone+\ldots+P^\un=(1,1)$. If we connect a line
between two consecutive points in the above sequence we get the \emph{sorted
path}, connecting $(0,0)$ with $(1,1)$.

A crucial observation is that all the $2^n$ possible solutions $P_I$ lie on or above
that sorted path. See Lemma \ref{lem:epigraph_integral} for a rigorous statement and proof and Figure \ref{fig:normal_dev_a_b_c200_m160_v6400} (right) for an example.
The fact that the \emph{sorted path always lies beneath all possible solutions} is \emph{independent of the target function}, and only depends on the fact that we deal with the normal distribution. Actually, this is true for any distribution where allocating a subset of services to a bin
amounts to adding the corresponding means and variances.

A second crucial observation, proved in Theorem \ref{thm:fractional}, is that the optimal fractional solution lies on the sorted path. In fact, the proof shows this is a general phenomenon that holds for many possible target functions. What we need from the cost function $Dev$ we are trying to minimize is:

\begin{itemize}
\item $Dev$ has a symmetry of reflection around a line that passes through a saddle point $(a',\frac{1}{2})$.

\item For every fixed $b$, $Dev$ is strictly uni-modal in $a$ (i.e., has a
unique minimum in $a$). We call the set of solutions
$\set{(m(b),b)}$ the \emph{valley}. The saddle point lies on the
valley.

\item $Dev$ restricted to the valley is strictly monotone for $b \le \half$ and $b \ge \half$ with a maximum at the saddle point that is obtained when $b=\half$.
\end{itemize}

%%\red{Later}
%\begin{figure}
%  % Requires \usepackage{graphicx}
%  \begin{center}
%  \includegraphics[width=6cm]{normal_dev_a_b_c200_m160_v6400_4_png}
%  \end{center}
%  \caption{$Dev(a,b)$ with the $2^{10}$ possible partitions represented in orange. The bottom sorted path (that sorts services by their VMR in \emph{increasing} order) and the upper sorted path (that sorts services by their VMR in \emph{decreasing} order), together with their integral points, are in black. Notice that all partition points are confined by both sorted paths.}
%  \label{fig:normal_dev_a_b_c200_m160_v6400_2}
%\end{figure}

To see why these conditions suffice, consider an arbitrary fractional solution $(a,b)$. By the
symmetry property we may assume without loss of generality that $b
\le \frac{1}{2}$. If $(a,b)$ is left to the valley, we can improve the
solution by keeping $b$ and moving $a$ towards the valley, until we
either hit the sorted path or the valley. If we hit the valley
first, we can improve the expected deviation, by going down the
valley until we hit the sorted path. A similar argument applies
when $(a,b)$ is right to the valley. The conclusion is that we can always find another fractional solution that is better than $(a,b)$ and lies on the sorted path.

Thus, quite surprisingly, we manage to decouple the question to two separate and almost orthogonal questions. The first question is what is the behavior of the function $Dev$ over its entire domain. Notice that $Dev$ is a function of only two variables, independent of $n$. This question \emph{does not depend on the input} $\set{(\mu^\ui,V^\ui)}$ at all. We study $Dev$ using analytic tools. The second question is what is the geometric structure of the set of feasible points, and here we investigate the input, \emph{completely ignoring the target function} $Dev$. For this question we use geometric intuition and combinatorial tools.

In Appendix \ref{sec:SP} we prove our three cost functions fall into the above framework. The discussion above also gives an intuitive geometric interpretation of the results in \cite{BE12,WMZ11}. We
believe the framework is also applicable to many other target functions.

There is an intuitive explanation why the optimal fractional solution lies on the sorted path. Sorting services by their VMR essentially sorts them by their risk, where risk is the amount of variance per one unit of expectation. Partitioning the sorted list corresponds to putting all the low risk services in one bin, and all the high risk services in the other. We call this the \emph{risk unbalancing principle}. Intuitively, we would like to give the high risk services as much spare capacity as possible, and we achieve that by grouping all low risk services together and giving them less spare capacity. In contrast, balancing risk amounts to taking the point $(\half,\half)$ (for the case where the two bins have equal capacities). Having the geometric picture in mind, we immediately see that this saddle point is not optimal, and can be improved by taking the point on the valley that intersects the sorted path.

The technique also applies to bins having different capacities that
previous work did not analyze, and reveals that we should allocate
low risk services to the bin with the lower capacity. This follows from a neat geometric argument that when $c_1<c_2$ the optimal fractional solution lies on the sorted path (that sorts services by their VMR in \emph{increasing} order) and not the upper sorted path (that sorts services by their VMR in \emph{decreasing} order). See Section \ref{sec:algorithm} and Theorem \ref{thm:fractional} for more details. If we have $k$ bins, we should sort the bins by their capacity
and the services by their VMR, and then allocate consecutive segments of the sorted list of services to the bins, with lower risk services allocated to smaller capacity bins. This double-sorting again intuitively follows from the risk unbalancing principle, trying to preserve as much spare capacity to the high risk services.\footnote{In a similar vein if we have total capacity $c$ to split between two bins, it is always better to make the bin capacities \emph{unbalanced} as much as possible, i.e., the minimum expected deviation decreases as $c_2 - c_1$ increases. In particular the \emph{best} choice is having a single bin with capacity $c$ and the \emph{worst} choice is splitting the capacities evenly between the two bins. Obviously, in practice there might be other reasons to have several bins, but if there is tolerance in each bin capacity it is always better to minimize the number of bins. We give a precise statement and proof in Appendix \ref{app:unbalance_bin_capacities}.}
We prove the double-sorting algorithm in Appendix \ref{sec:k}.

%We can also quite accurately estimate the difference between the
%expected deviation of the best \emph{integral} point on the sorted path and the optimal \emph{fractional} solution, using the
%mean value theorem. This can be easily done for SP-MOP and is a bit more involved for SP-MED,
%but the bottom line is that in both cases if we have $n$ services
%and none is too \emph{dominant} (meaning that all services have mean and variance considerably smaller than the \emph{total} mean and variance), then the error magnitude (as a fraction of
%the total capacity $\mu$) decreases linearly fast in $n$.  See Theorem \ref{thm:error} for exact details.
%%This gives a much better error estimate than \cite{BE12,WMZ11}.\footnote{\cite{BE12} give an approximation ratio $2$, and also an online competitive ratio close to $2$, that hold without the assumption that no service is too dominant. We show that the approximation ratio is very close to $1$ as long as no service is too dominant.}

Finally, we are left with the question of \emph{finding} the right
partition of the sorted list. With two bins
one can simply try all $n$ partition points. However, with $k$ bins
there are ${n \choose k-1} = \Theta(n^{k-1})$ possible partition points. We
show a dynamic programming algorithm that finds the best partition
in $poly(n)$ time. The optimal solution can probably be found much faster and we have candidate algorithms that work well in practice and we are currently working on formally proving their correctness.

\section{Formal treatment}

%\subsection{Problem Formulation}
%\label{sec:problem}

The input to the problem consists of $k$ and $n$, specifying the number of bins and services, integers $\set{c_i}_{i=1}^k$, specifying the bin capacities,
and values $\set{(\mu^\ui,V^\ui)}_{i=1}^n$, where the demand distribution $X^\ui$ of service $i$ is normal with mean $\mu^\ui$ and variance $V^\ui$. $X^\ui$ are independent.
The output is a partition of $[n]$ to $k$ disjoint sets $S_1,\ldots,S_k \subseteq [n]$, where $S_j$ includes indices of services that needs to be allocated to bin $j$. Our goal is to find a partition minimizing a given cost function
$D(S_1,\ldots,S_k)$. The optimal (integral) partition is the partition with minimum cost among all possible partitions.

We let $c$ denote the total capacity, $\mu$ the total demand and $V$ the total variance, i.e., $c=\sum_{j = 1}^{k} c_j$, $\mu=\sum_{i = 1}^{n} \mu^\ui$ and $V=\sum_{i = 1}^{n} V^\ui$. The value $c-\mu$ represents the total spare capacity we have. The (total) standard deviation is $\sigma=\sqrt{V}$ which represents the standard deviation of the input had all services been put into a single bin. We let $\Delta$ denote the spare capacity in units of the standard deviation, i.e., $\Delta = \frac{c - \mu}{\sigma}$. We assume that $c \geq \mu$.

We remind the reader that the sum of independent normal distributions with mean $\mu^\ui$ and variance $V^\ui$ is normal with mean $\sum \mu^\ui$ and variance $\sum V^\ui$. Consider a partition $S_1,\ldots,S_k$. Then, the demand distribution of bin $j$ is normal with mean $\mu_j=\sum_{i \in S_j} \mu^\ui$ and variance $V_j=\sum_{i \in S_j} V^\ui$. The standard deviation of bin $j$ is $\sigma_j=\sqrt{V_j}$ and its spare capacity in units of its standard deviation is $\Delta_j = \frac{c_j - \mu_j}{\sigma_j}$. Following previous work, we assume that for every $i$, $\mu^\ui \ge 0$ and $V^\ui$ is small enough compared with $(\mu^\ui)^2$ so that the probability of getting negative demand in service $i$ is negligible.

Next we normalize everything with regard to the total mean $\mu$ and the total variance $V$. The input to the function $D$ are two vectors $a_1,\ldots,a_{k-1},b_1,\ldots,b_{k-1}$ s.t., $\sum_{j=1}^{k-1} a_j \le 1$ and $\sum_{j=1}^{k-1} b_j \le 1$. We think of the vectors  $a_1,\ldots,a_{k-1}$ and $b_1,\ldots,b_{k-1}$ as representing the fraction of mean and variance each bin takes. We let $a_k=1-\sum_{j=1}^{k-1} a_j$ and $b_k=1-\sum_{j=1}^{k-1} b_j$ be the fraction of the last bin (that takes all the remaining mean and variance). We let $\mu_j=a_j \mu$ and $V_j=b_j V$ for $j=1,\ldots,k$ and we define $\sigma_j=\sqrt{V_j}=\sqrt{b_j}\sigma$ (where $\sigma=\sqrt{V}$) and $\Delta_j=\frac{c_j-\mu_j}{\sigma_j}$. In this notation the cost function is $D(a_1,\ldots,a_{k-1};b_1,\ldots,b_{k-1})$.

The function $D$ is defined over all tuples $a_1,\ldots,a_{k-1}$, \allowbreak $b_1,\ldots,b_{k-1} \in [0,1]$ s.t., $\sum_{j=1}^{k-1} a_j \le 1$ and $\sum_{j=1}^{k-1} b_j \le 1$. There are $k^n$ possible partitions that correspond to the $k^n$ (possibly) different tuple values. We call a tuple \emph{integral} if it is induced by some partition. Our goal is to approximate the integral solution minimizing $D$ over all integral inputs.

\subsection{What do we require form a cost function?}
\label{sec:cost}

We can handle a wide variety of cost functions. Specifically, we require the following from the cost function:

\begin{enumerate}
\item
\label{it:kto2}
In any solution to a $k$-bin problem, the allocation for any two bins is also optimal. Formally, if $S_1,\ldots,S_k$ is optimal for a $k$-bin problem, then
for any $j$ and $j'$, the partition $S_j,S_{j'}$ is optimal for the two-bin problem defined by the services in $S_j \cup S_{j'}$ and capacities $c_j, c_{j'}$. We remark that this is a natural condition that is true for almost any cost function we know.
\end{enumerate}

\noindent
For $k=2$ we require that:

\begin{enumerate}
\setcounter{enumi}{1}
%\item
%\label{it:diffrentiable} $D(a,b)$ is differentiable in the range
%$[0,1] \times [0,1]$.
% This condition is false for SPMWOP.

\item
\label{it:symmetry}
$D(a,b)=D(1-a-\frac{c_2-c_1}{\mu},1-b)$. When $c_1=c_2$ this simply translates to $D(a,b)=D(1-a,1-b)$ and there is no difference between allocating the set $S_1$ to the first bin or to the second one.

\item
\label{it:valley} For every fixed $b \in [0,1]$, $D(a,b)$ has a
unique minimum on $a \in [0,1]$, at some point $a=m(b)$, i.e. $D$ is
decreasing at $a < m(b)$ and increasing at $a>m(b)$. We call the
points on the curve $\set{(m(b),b)}$ the \emph{valley}.

\item
\label{it:saddle} $D$ has a unique maximum over the valley at the
point $(m(\frac{1}{2}),\frac{1}{2})$. In fact by the symmetry above
this point is $(\half-\frac{c_2-c_1}{2\mu},\half)$. This means that
$D(m(b),b)$ is increasing for $b \le \half$ and decreasing for $b
\ge \half$.
\end{enumerate}

\subsection{Three cost functions}
We have already seen SP-MED, which minimizes the expected deviation. Previous work on SBP (stochastic bin packing) studied the worst overflow probability of a bin. We define SP-MWOP to be the problem that minimizes the worst overflow probability of a bin. Another natural problem is what we call SP-MOP that minimizes the probability that some bin overflows. The three cost functions are related but behave differently. For example, SP-MWOP is not differentiable but the other two are. It is relatively simple to find the valley in SP-MED and SP-MWOP but we are not aware of any explicit description of the valley in SP-MOP.

Nevertheless, it is remarkable that all three cost functions fall into our framework and we prove that in Appendix \ref{sec:SP}. The proof is not always simple, and we have used Lagrange multipliers and the log concavity of the cumulative distribution function of the normal distribution for proving this for SP-MOP. It follows that all the machinery we develop in the paper holds for these cost functions.

\section{The two bin case}
\label{sec:SPMED:2}

We first consider the two bin case ($k=2$). We shall later see (in Appendix \ref{sec:k}) that solving the $k=2$ case is key for solving the general problem for any $k > 2$. As explained before we decouple the question to one about the function $D$ and another about the structure of the integral points.

\subsection{The sorted path}
\label{sec:2:sorted_path}

We now consider the input
$(\mu^\uone,V^\uone),\ldots,(\mu^\un,V^\un)$. We represent service
$i$ with the pair $(a^\ui=\frac{\mu^\ui}{\mu},b^\ui=\frac{V^\ui}{V})$. If we split the input to the sets
$I$ and $[n] \setminus I$, then the first bin is normally distributed with mean $\mu \sum_{i \in I} a^\ui$ and variance $V  \sum_{i \in I} b^\ui$.
Thus, the $n$ input points induce $2^n$ possible solutions $P_I=(\sum_{i \in I} a^\ui,\sum_{i
\in I} b^\ui)$ for each $I \subseteq [n]$ and we call each such point \emph{an integral point}.
Sorting the services by their VMR, is equivalent to sorting the vectors $P^\ui=(a^\ui,b^\ui)$ by the
angle they make with the $a$ axis.

\begin{definition}(The sorted path)
Sort the services by their VMR in increasing order and calculate the $P^\uone,P^\utwo,\ldots,P^\un$ vectors. For $i=1,\ldots,n$ define
\begin{eqnarray*}
P_{bottom}^{[i]} &=& P^\uone+P^\utwo+\ldots+P^\ui \mbox{ and,}\\
P_{up}^{[i]} &=& P^\un+P^{(n-1)}+\ldots+P^{(n-i+1)},
\end{eqnarray*}
and also define $P_{bottom}^{[0]}=P_{up}^{[0]}=(0,0)$.
%Note that $P_{bottom}^{[n]}=P_{up}^{[n]}=(1,1)$.

The \emph{bottom sorted path} is the curve that is formed by connecting $P_{bottom}^{[i]}$ and $P_{bottom}^{[i+1]}$ with a line, for $i=0,\ldots,n-1$.
The \emph{upper sorted path} is the curve that is formed by connecting $P_{up}^{[i]}$ and $P_{up}^{[i+1]}$ with a line, for $i=0,\ldots,n-1$. We sometimes abbreviate the bottom sorted path and call it \emph{the sorted path}.
\end{definition}

The integral point $P_{bottom}^{[i]}$ on the bottom sorted path corresponds to allocating the $i$ services with the lowest VMR to the first bin and the rest to the second. On the other hand, the integral point $P_{up}^{[i]}$ on the upper sorted path corresponds to allocating the $i$ services with the highest VMR to the first bin and the rest to the second. A crucial, yet simple, observation (proven in Appendix \ref{app:epigraph_integral_proof}):

\begin{lemma}
\label{lem:epigraph_integral}
All the integral points lie within the polygon confined by the bottom sorted path and the upper sorted path.
\end{lemma}

Lemma \ref{lem:epigraph_integral} is a key lemma of the paper. Unfortunately, we had to omit the proof because of lack of space. We recommend the reader to look at the simple proof that appears in Appendix \ref{app:epigraph_integral_proof}. In fact, 

%. A fractional partition is one that allows splitting a service between several bins. Geometrically, the set of fractional points is a convex set. Clearly, it contains all the points on both the bottom sorted path and the upper sorted path, and because it defines a convex set, also all points in their convex hull. In fact,

\begin{lemma}
\label{lem:epigraph_fractional}
The set of fractional points coincides with the polygon confined by the bottom sorted path and the upper sorted path.
\end{lemma}

\subsection{The sorting algorithm}
\label{sec:algorithm}

We now present the algorithm for the $k=2$ case. The algorithm for the general case of arbitrary $k$ is presented in Appendix \ref{sec:k}.
The algorithm gets as input $n,\set{(\mu^\ui,V^\ui)}_{i=1}^n,c_1,c_2$ and outputs a partition $S_1,S_2$ of $[n]$ minimizing cost. We keep all notation as before. The algorithm works as follows:

\begin{boxit}
\bcaption{The sorting algorithm: $2$ bins}
\begin{itemize}
\item Sort the bins by their capacity such that $c_1 \le c_2$.
\item Sort the services by their VMR such that $\frac{v^\uone}{\mu^\uone} \le \frac{v^\utwo}{\mu^\utwo} \le \dots \le \frac{v^\un}{\mu^\un}$.
\item Calculate the points $P^{[0]}=(0,0),P^{[1]},$ $\ldots,P^{[n]}=(1,1)$ on the sorted path.
\item Calculate $D(P^{[i]})$ for each $0 \le i \le n$ and find the index $i^*$ such that the point $P^{[i^*]}$ achieves the minimal cost among all points $P^{[i]}$.
\item Let $S_1=\set{1,\ldots,i^*}$ and $S_2=[n] \setminus S_1$. Output $(S_1,S_2)$.
\end{itemize}
\end{boxit}

The optimal fractional point is the fractional point that minimizes cost. In Section \ref{sec:risk_unbalancing_principle} we said, and soon will prove, that the optimal solution allocates low risk services to one bin and the rest to the other. However, offhand, it is not clear whether to allocate the smaller risk services to the lower capacity bin or the higher capacity bin. In fact, offhand, it is not clear whether the optimal solution is on the lower sorted path or the upper sorted path, and it might even depend on the input. Remarkably, we prove that it always lies on the bottom sorted path, meaning that it is always better to allocate low risk services to the smaller capacity bin and high risk services to the higher capacity bin.  We gave an intuitive explanation to this phenomenon in Section \ref{sec:risk_unbalancing_principle}. We prove Theorem \ref{thm:fractional} in Appendix \ref{app:fractional_proof}.

\begin{theorem}
\label{thm:fractional}
The optimal fractional point lies on the bottom sorted path. The optimal fractional solution splits at most one service between two bins.
\end{theorem}

\section{Simulation results}

In this section we present our simulation results for the two bin
and $k$ bin cases. We run our simulation on bins with equal
capacity. We compare the sorting algorithms for SP-MED, SP-MWOP and SP-MOP to
an algorithm we call BM (Balanced Mean). The algorithm goes through
the list, item by item, and allocates each item to the bin which is
less occupied, which is a natural benchmark and also much better
than other naive solutions like first-fit and first-fit decreasing.
\footnote{At first, we also wanted to compare our algorithm with
variants of the algorithms considered in \cite{WMZ11,BE12} for the
SBP problem. In both papers, the authors consider the algorithms
First Fit and First Fit Decreasing \cite{RGSJ79} with item  size
equal to the effective size, which is the mean value of the item
plus an extra value that guarantees an overflow probability be at
most some given value $p$. Their algorithm chooses an existing bin
when possible, and otherwise opens a new bin. However, when the
number of bins is fixed in advance, taking effective size rather
than size does not change much. For a new item (regardless of its
size or effective size) we keep choosing the bin that is less
occupied, but this time we measure occupancy with respect to
effective size rather than size. Thus, if elements come in a random
order, the net outcome of this is that the two bins are almost
balanced and a new item is placed in each bin with almost equal
probability.}

We show simulation results on synthetic data. We first generate our
stochastic input $\set{(\wt{\mu}^\ui,\wt{\sigma}^\ui)}_{i=1}^n$ for
$n=100$ and for $n=500$. Our sample space is a mixture of three populations: all
items have the same mean (we fixed it at $\wt{\mu}^\ui=500$) but
50\% had standard deviation picked uniformly from $[0, 0.1 \cdot
\widetilde{\mu}^\ui]$, 25\% had standard deviation picked uniformly from $[0.1 \cdot
\widetilde{\mu}^\ui, 0.5 \cdot \widetilde{\mu}^\ui]$ and 25\% had standard deviation picked
uniformly from $[0.5 \cdot \widetilde{\mu}^\ui, 0.9 \cdot
\widetilde{\mu}^\ui]$.

We then randomly generated $500$ sample values $x_l^\ui$ for each $1
\le i \le n$ and $1 \le l \le 500$ using the normal distribution
$N[\widetilde{\mu}^\ui,\widetilde{V}^\ui]$ and from this we inferred
parameters $\mu^\ui,V^\ui$, best explaining the sample as a normal
distribution. Both the sorting algorithm and the BM algorithm got as
input $\set{(\mu^\ui,V^\ui)}_{i=1}^n$, as well as
$c_1=\ldots=c_k=\frac{c}{k}$ and output their partition.

To check the suggested partitions, we viewed each sample $x_l^\ui$
as representing an item instantiation in a different time slot. We
then computed the cost function. For example, for SP-MED, the
deviation value for bin $j$ at time slot $l$ is: $\max \set{0, 100
\frac{\sum_{i \in S_j}x_l^\ui - c_j}{\sum_{i=1}^n \mu^\ui}}$, i.e.,
the deviation is measured as a percent of the total mean value
$\mu$. We generated $20$ such lists and calculated the average cost
for these $20$ input lists for each algorithm. We run the experiment
for different values of $c$.

Figure \ref{fig:2bins_cost} shows the $2$ bins
average cost of both algorithms for SP-MED, SP-MWOP and SP-MOP as a function of
$\frac{c}{\mu}$. As expected, the average
cost decreases as the value $\frac{c}{\mu}$ increases, i.e. as
the total spare capacity increases. We also see that the sorting
algorithm out-performs BM. Figure \ref{fig:2bins_cost_ratio}
shows the average cost of the BM algorithm divided by the average
cost of the sorting algorithm for the three problems, again as a function of
$\frac{c}{\mu}$.

\begin{figure}[t!]
  \begin{center}
  \includegraphics[width=4cm]{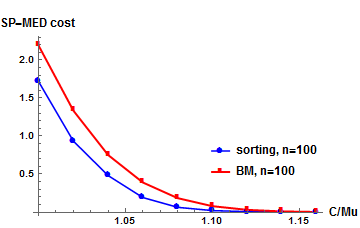}
  \includegraphics[width=4cm]{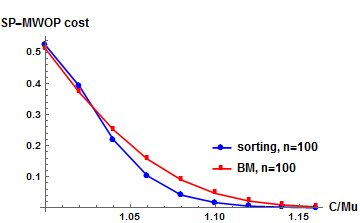}
  \includegraphics[width=4cm]{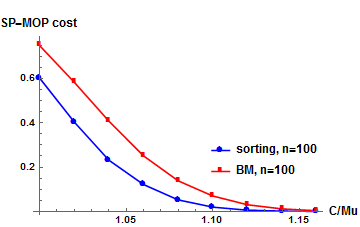}
  \end{center}
  \caption{Average cost of the sorting algorithm and the BM algorithm for SP-MED, SP-MWOP and SP-MOP with two bins. The $x$ axis measures $\frac{c}{\mu}$.}
  \label{fig:2bins_cost}
\end{figure}

\begin{figure}[t!]
  \begin{center}
  \includegraphics[width=4cm]{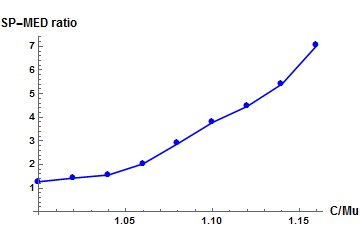}
  \includegraphics[width=4cm]{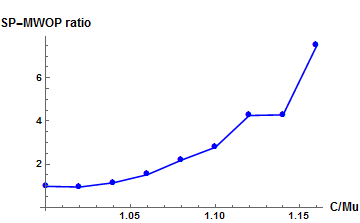}
  \includegraphics[width=4cm]{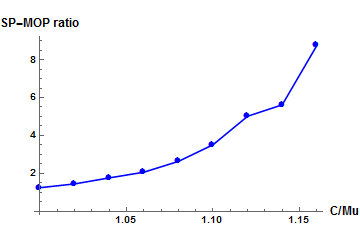}
  \end{center}
  \caption{Average cost of the BM algorithm divided by average cost of the sorting algorithm for SP-MED, SP-MWOP and SP-MOP with $2$ bins. The $x$ axis measures $\frac{c}{\mu}$.}
  \label{fig:2bins_cost_ratio}
\end{figure}

Figure \ref{fig:4bins} shows the $4$ bins
average cost and cost ratio of both algorithms for SP-MED as a function of
$\frac{c}{\mu}$ for $n=100$ and $n=500$. As we saw in the two bins case, the average
cost decreases as the value $\frac{c}{\mu}$ increases. We can also see that the average deviation for $n=100$ is higher than the deviation for $n=500$.
However, the ratio between the BM algorithm average cost and the sorting algorithms average cost for $n=100$ is lower than the ratio we get for $n=500$. Still, in both $n$ values, the sorting algorithm out-performs the BM.

\begin{figure}[t!]
  \begin{center}
  \includegraphics[width=5cm]{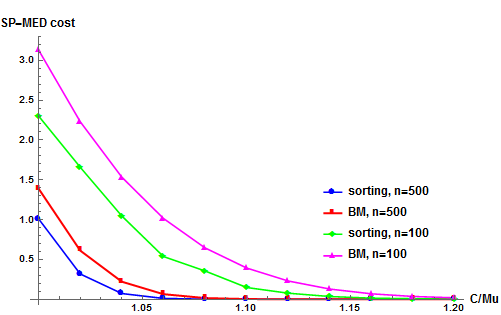}
  \includegraphics[width=5cm]{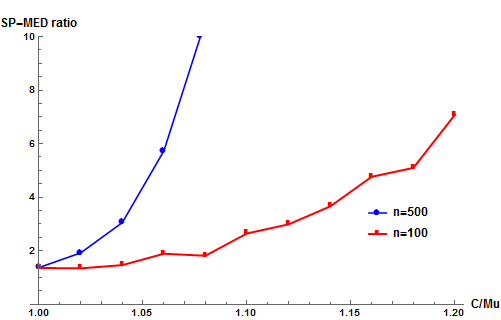}
  \end{center}
  \caption{The left figure depicts the average cost of the sorting algorithm and the BM algorithm for SP-MED with four bins, both for $n=100$ and for $n=500$. The right figure depicts the average cost of the BM algorithm divided by the average cost of the sorting algorithm for $n=100$ and for $n=500$. The $x$ axis measures $\frac{c}{\mu}$.}
  \label{fig:4bins}
\end{figure}

\bibliography{BibFile}

\appendix
 
 \section{The general case}
\label{sec:k}
We now analyze the general $k$ bin case using the
results we have obtained for the two-bin case. We first show the
optimal \emph{factional} solution is obtained by first sorting the
bins by their capacities and the services by their VMR ratio, and
then partitioning the sorted list to $k$ consecutive parts.

Assume the bins are sorted by their capacity, $c_1 \le c_2 \le \ldots \le c_k$.
A fractional solution is called \emph{sorted} if for every $j < j'$ and every two services $i$ and $i'$ such that service $i$ (resp. $i'$) is allocated to bin $j$ (resp. $j'$) it holds that the VMR of service $i$ is at most that of service $i'$.

\begin{theorem}
\label{thm:k:sorted}
The optimal fractional solution is \emph{sorted}.
\end{theorem}

\begin{proof}
Assume there is an optimal solution $S_1,\ldots,S_k$ that is \emph{not} sorted. I.e., there exist $j < j'$ and $i,i'$ such that service $i$ (resp. $i'$) is allocated to bin $j$ (resp. $j'$) and the VMR of service $i$ is strictly larger than that of service $i'$. By Theorem \ref{thm:fractional} the sorted fractional solution for the problem defined by bins $j$ and $j'$ is \emph{strictly better} than the one offered by the solution $S_1,\ldots,S_k$. This means the solution is not optimal for bins $j$ and $j'$ - in contradiction to item (\ref{it:kto2}) in Section \ref{sec:cost}.
\end{proof}

The fact that with two bins and different capacities $c_1 < c_2$, low risk services should be allocated to the smaller capacity bin (See Theorem \ref{thm:fractional} and the discussion before it) implies also that for $k$ bins lower risk services should be allocated to lower capacity bins.

Next we present an algorithmic framework for the problem:

\begin{boxit}
\bcaption{The double sorting algorithm framework}
\begin{itemize}
\item Sort the bins by their capacity $c_1 \le c_2 \le \ldots \le c_k$.
\item Sort the services by their VMR $\frac{v^\uone}{\mu^\uone} \le \frac{v^\utwo}{\mu^\utwo} \le \dots \le \frac{v^\un}{\mu^\un}$.
\item Use a partitioning algorithm to find $k-1$ partition points $\ell_0=0 \le \ell_1 < \ldots \ell_{k-1} \le \ell_k=n$ and output $S_1,\ldots,S_k$ where $S_j=\set{\ell_{j-1}+1,\ldots,\ell_j}$.
\item Allocate the services in $S_j$ to bin number $j$ (with capacity $c_j$).
\end{itemize}
\end{boxit}

The algorithmic framework is not complete since it does not specify
how to find the partition points on the sorted path. With two bins
there is only one partition point, and we can check all possible
$n-1$ integral partition points. With $k$ bins there are ${n \choose
k-1}$ possible partition points, and checking all partition points
may be infeasible.

We now describe a dynamic programming algorithm that solves the
problem, and works under very general conditions.

\begin{boxit}
\bcaption{Finding an optimal integral partition}
\mbox{ }

For each $1 \le t \le \log k$, $1 \le i < i' \le n$ and $j=1,\ldots,k$ keep the two values
$Partition(2^t,i,i',j)$ and $D(2^t,i,i',j)$ that are defined as follows:\\

\begin{itemize}
\item Base case $t=1$: $Partition(2,i,i',j)$ is the best integral partition point for the two bin problem with inputs $X_i,\ldots,X_{i'}$ and capacities $c_j,c_{j+1}$. $D(2,i,i',j)$ is its corresponding expected deviation.\\

\item Induction step. Suppose we have built the tables for $t$, we show how to build the tables for $t+1$. We let $Dev(2^{t+1},i,i',j)$ be
$$
\min_{i \le i'' \le i'}
\set{D(2^t,i,i''-1,j) + D(2^t,i'',i',j+2^t)}
$$
and we let $Partition(2^{t+1},i,i',j)$ be the partition point that obtains the minimum in the above equation.
\end{itemize}
\mbox{ }

The optimal partition to $k$ bins (assuming $k$ is a power of two) can be recovered from the tables. For example for $4$ bins $Partition(4,1,n,1)$ returns the middle point $\ell_2$ of the partition, and the partition points within each half are obtained by $\ell_1=Partition(2,1,\ell_2-1,1)$ and  $\ell_3=Partition(2,\ell_2,n,3)$.
\end{boxit}

It can be seen (by a simple induction) that $Partition(2^t,i,i',j)$ gives the middle point of the optimal integral solution on the sorted path to the problem of dividing $X_i,\ldots,X_{i'}$ to $2^t$ bins with capacities $c_j,c_{j+1},\ldots,c_{j+2^t-1}$. It can also be verified that the running time of the algorithm is $O(n^3 k \log k)$.

We analyze the error the integral solution we find has compared with the optimal fractional solution. The error analysis builds upon the error analysis of the two bin case. Start with an optimal fractional solution $\pi$. By Theorem \ref{thm:k:sorted} we know $\pi$ is sorted. $\pi$ defines some sorted fractional solution to the problem defined by the first two bins. By the remark after Theorem \ref{thm:error} we know that one of the integral points to the left or to the right of the fractional solution has a very close deviation to that of $\pi$, and we fix the first partition point accordingly to make it integral. Now assume after fixing $j$ stages we hold some fractional solution $\pi^\uj$ that is integral on the first $j$ partition points and possibly fractional on the rest. We look at bins $j+1$ and $j+2$ and find an integral partition point (to the left or to the right of the next fractional partition point) that is almost as good as the fractional one. Doing the above for $j=1,\ldots,k-1$ we end up with an integral solution whose deviation is close to that of the optimal fractional solution, where the error bound in Theorem \ref{thm:error} is multiplied by a factor of $k-1$. We remark that this also shows an efficient way of converting an optimal fractional solution to an integral solution close to it.

The algorithm is quite general and works whenever we can solve the $k=2$ case, and therefore may be seen as a reduction from arbitrary $k$ to the $k=2$ case. The main point we want to stress is that even in this generality the complexity is polynomial. The main disadvantage of the algorithm is that it runs in about $n^3$ time. For specific cases (such as SP-MED and SP-MOP and probably other natural variants) the optimal solution may be probably found much faster,
and we are currently working on proving the correctness of an almost linear algorithm.

\section{Three cost functions}
\label{sec:SP}

\subsection{Stochastic Packing with Minimum Expected Deviation (SP-MED)}

In Appendix \ref{app:spmed:dev_single_bin} we prove the expected
deviation of bin $j$ ($j=1,\dots,k$), denoted by $Dev_{S_j}$, is
$$Dev_{S_j}=
\sigma_j [\phi(\Delta_j) - \Delta_j (1-\Phi(\Delta_j))],$$ where
$\phi$ is the probability density function (pdf) of the standard
normal distribution and $\Phi$ is its cumulative distribution
function (CDF).\footnote{A quick background on the normal
distribution is given in Appendix
\ref{app:standard_normal_distribution}.} Denoting  $g(\Delta) =
\phi(\Delta) - \Delta (1-\Phi(\Delta) )$ we see that
$Dev_{S_j}=\sigma_j~ g(\Delta_j)$. With two bins $Dev$ is a function
from $[0,1]^2$ to $\R$ and $Dev(a,b)=\sigma_1 g(\Delta_1)+ \sigma_2
g(\Delta_2)$ where the first bin has mean $a\mu$ and variance $bV$,
the second bin has mean $(1-a)\mu$ and variance $(1-b)V$ and
$\sigma_j,\Delta_j$ are defined as above.

\begin{lemma}
$Dev$ respects conditions (\ref{it:kto2})-(\ref{it:saddle}) of
Section \ref{sec:cost}.
\end{lemma}

\begin{proof}
We go over the conditions one by one.
\begin{description}
\item [item (\ref{it:kto2}):] Let $S_1,\ldots,S_k$ be an optimal solution. Suppose there are
$j$ and $j'$ for which $S_j,S_{j'}$ are not the optimal solution for
the two-bin problem defined by the services in $S_j \cup S_{j'}$ and
capacities $c_j, c_{j'}$. Change the allocation of the solution
$S_1,\ldots,S_k$ on bins $j$ and $j'$ to an optimal solution for the
two bin problem. This change improves the expected total deviation
of the two bins while not affecting the expected deviation of any
other bin. In total we get a better solution, in contradiction to
the optimality of $S_1,\ldots,S_k$.

%\item [item (\ref{it:diffrentiable}):] Clearly $Dev$ is
%differentiable.

\item [item (\ref{it:symmetry}):] In Appendix
\ref{app:spmed:symmetry} we prove
$Dev(a,b)=Dev(1-a-\frac{c_2-c_1}{\mu}, 1-b)$. We remark that for
$c_1=c_2$ this simply says we can switch the names of the first and
second bin.

\item [item (\ref{it:valley}):] In Appendix \ref{app:spmed:partial_a_derivative} we
calculate $\frac{\partial^2 Dev}{\partial a^2} = \mu^2~
[\frac{\phi(\Delta_2)}{\sigma_2} + \frac{\phi(\Delta_1)}{\sigma_1}]
\ge 0$. It follows that for any $0 < b < 1$, $Dev(a)$ is convex and
has a unique minimum. The unique point $(m(b),b)$ on the valley is
the one where $\Delta_1 = \Delta_2$.

\item [item (\ref{it:saddle}):] We first explicitly determine what
$Dev$ restricted to the valley is as a function $D(b)=Dev(m(b),b)$
of $b$. As $Dev(a,b)=\sigma_1 g(\Delta_1)+\sigma_2 g(\Delta_2)$ and
on the valley $\Delta_1=\Delta_2$ we see that on the valley
$Dev(a,b)=(\sigma_1+\sigma_2) g(\Delta_1)$. However, $\sigma_1 +
\sigma_2$ also simplifies to
$\frac{c_1-a\mu}{\Delta_1}+\frac{c_2-(1-a)\mu}{\Delta_2}=\frac{c-\mu}{\Delta_1}$.
Altogether, we conclude that on the valley $Dev(a,b)=(c-\mu)
\frac{g(\Delta_1)}{\Delta_1}$ is a function of $\Delta_1$ alone.

Now it is a straight forward calculation that $\frac{\partial
Dev(\Delta_1)}{\partial \Delta_1} = - (c-\mu)~
\frac{\phi(\Delta_1)}{\Delta_1^2} < 0$.
%\begin{eqnarray*}
%\frac{\partial Dev(\Delta_1)}{\partial \Delta_1}
%& = &
%(C-D)~ \frac{\frac{\partial g(\Delta_1)}{\partial \Delta_1} \Delta_1 - g(\Delta_1)}{\Delta_1^2} \\
%& = &
%(C-D)~ \frac{[\Phi(\Delta_1) - 1] \Delta_1 - [\phi(\Delta_1) + \Delta_1 \Phi(\Delta_1) - \Delta_1]}{\Delta_1^2} \\
%& = &
%(C-D)~ \frac{\Delta_1 \Phi(\Delta_1) - \Delta_1 - \phi(\Delta_1) - \Delta_1 \Phi(\Delta_1) + \Delta_1}{\Delta_1^2} \\
%& = &
%- (C-D)~ \frac{\phi(\Delta_1)}{\Delta_1^2}
%\end{eqnarray*}
Also $\frac{\partial \Delta_1}{\partial b}$ is negative when $b \le
\half$ and positive when $b \ge \half$ (see Appendix
\ref{app:spmed:partial_a_derivative} and
\ref{app:spmed:partial_b_derivative}). As $\frac{\partial
D}{\partial b}=\frac{\partial Dev}{\partial \Delta_1} \cdot
\frac{\partial \Delta_1}{\partial b}$, we see that $D(b)$ is
increasing for $b \le \half$ and decreasing for $b \ge \half$ as
claimed. The saddle point $(a= \half-\frac{c_2-c_1}{2\mu},
b=\frac{1}{2})$ lies on the valley and is a maximum point for $Dev$
restricted to the valley.

We remark that we could simplify the proof by using Lagrange
multipliers (as we do in Lemma \ref{lem:lagrange}). However, since here it
is easy to explicitly find $Dev$ restricted to the valley we prefer
the explicit solution. Later, we will not be able to explicitly find
the restriction to the valley and we use instead Lagrange
multipliers that solves the problem with an \emph{implicit}
description of the valley.
\end{description}
\end{proof}

%The two valleys do not reach the points $(0,0)$ and $(1,1)$. Rather,
%for $b=1$ we get $a=a(1)=1-\frac{c_2}{\mu}$ and for $b=0$ we get
%$a=a(0)=1-\frac{c_1}{\mu}$. These points are the absolute minimum of
%the function $Dev$ in the range $[0,1]^2$. In these points one bin
%has variance $0$ and is full to its capacity, while the other takes
%variance $V$ and has maximal spare capacity $c-\mu$.

\begin{lemma}
\label{lem:spmed:error} The difference between the expected
deviation in the integral point found by the sorting algorithm and
the optimal integral (or fractional) point for SP-MED is at most
$\frac{1}{\sqrt{2 \pi e}} \cdot \frac{1}{\alpha} \cdot \frac{L}{n}
\cdot \mu$. In particular, when $L=\littleo{n}$ and $\alpha$ is a
constant, the error is $\littleo{\mu}$.
\end{lemma}

\begin{proof}
We know from Theorem \ref{thm:error} that the difference is at most
$$
\min \set{|\nabla D(\xi_1)| , |\nabla D(\xi_2)|} \frac{L}{n},$$
where $\xi_1 \in [O_1,OPT_f]$, $\xi_2 \in [OPT_f,O_2]$ and $O_1$ and
$O_2$ are the two points on the bottom sorted path between which
$OPT_f$ lies. Using the partial derivatives calculated in Appendix
\ref{app:spmed:partial_derivatives} we see that

\begin{eqnarray*}
| \nabla (\sigma_2 g)(\Delta_2) | & \leq & |\mu~ (1 -
\Phi(\Delta_2))| + | \frac{\sigma}{2 \sqrt{1-b}}~ \phi(\Delta_2) |
 ~ \le ~ \mu+\frac{\sigma}{2 \sqrt{1-b}}~ \phi(\Delta_2).
\end{eqnarray*}

Moreover, $\frac{\sigma}{2 \sqrt{1-b}}~
\phi(\Delta_2)=\frac{\sigma}{2 \sqrt{1-b}}~ \frac{1}{\Delta_2}~
\Delta_2 \phi(\Delta_2)$ and a simple calculation shows that the
function $\Delta \phi(\Delta)$ maximizes at $\Delta=1$ with value at
most $\frac{1}{\sqrt{2 \pi e}}$. By our assumption that $\Delta_j
\ge 0$ for every $j$, we get that

\begin{eqnarray*}
\frac{\sigma}{2 \sqrt{1-b}}~ \phi(\Delta_2) & \leq & \frac{\sigma}{2
\sqrt{1-b}}~ \frac{\sigma \sqrt{1-b}}{c_2 -(1-a)\mu}~
\frac{1}{\sqrt{2 \pi e}}  ~ \leq ~
\label{eqn:b2e} % Bin 2 Error
\frac{V}{2\sqrt{2 \pi e}} ~ \frac{1}{c_2 - (1-a)\mu}.
\end{eqnarray*}

Applying the same argument on $O_2$ shows the error can also be
bounded by $\frac{V}{2\sqrt{2 \pi e}} ~ \frac{1}{c_1 - a\mu}$.

However, $(c_1 - a\mu) + (c_2 - (1-a)\mu) = c - \mu$ which is the
total spare capacity, and at least one of the bins takes spare
capacity that is at least half of that, namely $\frac{c-\mu}{2}$.
Since the error is bounded by either term, we can choose the one
where the spare capacity is at least $\frac{c-\mu}{2}$ and we
therefore see that the error is at most $\frac{V}{2\sqrt{2 \pi e}} ~
\frac{2}{c-\mu}$. Since we assume $c-\mu \ge \alpha \mu$ for some
constant $\alpha>0$, the error is at most $\frac{V}{\sqrt{2 \pi e}}
~ \frac{1}{\alpha \mu}$. As we assume $V \le \mu^2$, $\frac{V}{\mu}
\le \mu$ which completes the proof.
\end{proof}

This shows the approximation factor goes to $1$ and linearly (in the
number of services) fast. Thus, from a practical point of view the
theorem is very satisfying.

\subsection{Stochastic Packing with Min Worst Overflow Probability (SP-MWOP)}

The SP-MWOP problem gets as input integers $k$ and $n$, specifying
the number of bins and services, integers $c_1, \ldots , c_k$,
specifying the bin capacities and values
$\set{(\mu^\ui,V^\ui)}_{i=1}^n$, specifying that the demand
distribution $X^\ui$ of service $i$ is normal with mean $\mu^\ui$
and variance $V^\ui$. A solution to the problem is a partition of
$[n]$ to $k$ disjoint sets $S_1,\ldots,S_k \subseteq [n]$ that
minimizes the worst overflow probability.

The SP-MWOP problem is a natural variant of SBP. For a given
partition let $OFP_j$ (for $j=1,\ldots,k$) denote the overflow
probability of bin $j$. Let $WOFP$ denote  the \emph{worst} overflow
probability, i.e., $WOFP = \max_{j=1}^k \set{OFP_j}$. In the SBP
problem we are given $n$ normal distributions and wish to pack them
into few bins such that the $OFP \le k$ for some given parameter
$p$. Suppose we solve the SBP problem for a given $p$ and know that
$k$ bins suffice. We now ask ourselves what is the minimal $WOFP$
achieved with the $k$ bins (this probability is clearly at most $p$
but can also be significantly smaller). We also ask what is the
partition that achieves this minimal worst overflow probability. The
problem SP-MOP does exactly that.

In Appendix \ref{app:spmwop:single_bin} we prove the overflow
probability of bin $j$ ($j=1,\dots,k$), is $OFP_j=1-\Phi(\Delta_j)$
where $\Delta_j=\frac{c_j - \mu_j}{\sigma_j}$. Thus, $$WOFP  =
\max_{j=1}^k \set{1-\Phi(\Delta_j)}.$$  With two bins $WOFP$ is a
function from $[0,1]^2$ to $\R$ and $WOFP(a,b)=\max
\set{1-\Phi(\Delta_1),1-\Phi(\Delta_2)}$ where the first bin has
mean $a\mu$ and variance $bV$, the second bin has mean $(1-a)\mu$
and variance $(1-b)V$ and $\sigma_j,\Delta_j$ are defined as above.

\begin{lemma}
$WOFP$ respects conditions (\ref{it:kto2})-(\ref{it:saddle}) of
Section \ref{sec:cost}.
\end{lemma}

\begin{proof}
We go over the conditions one by one.
\begin{description}
\item [item (\ref{it:kto2}):] Let $S_1,\ldots,S_k$ be an optimal solution. Suppose there are
$j$ and $j'$ for which $S_j,S_{j'}$ are not the optimal solution for
the two-bin problem defined by the services in $S_j \cup S_{j'}$ and
capacities $c_j, c_{j'}$. Change the allocation of the solution
$S_1,\ldots,S_k$ on bins $j$ and $j'$ to an optimal solution for the
two bin problem. This change improves the worst overflow probability
of the two bins while not affecting the overflow probability of any
other bin. In total we get a better solution, in contradiction to
the optimality of $S_1,\ldots,S_k$.

%\item [item (\ref{it:diffrentiable}):] Clearly $WOFP$ is **not**
%differentiable.

\item [item (\ref{it:symmetry}):] The same proof as in Appendix
\ref{app:spmed:symmetry} shows
$WOFP(a,b)=WOFP(1-a-\frac{c_2-c_1}{\mu}, 1-b)$.

\item [item (\ref{it:valley}):]
Fix $b$. Denote $OFP_{1}(a,b)=OFP_{S_1}(a\mu,bV)=1-\Phi(\Delta_1)$.
It is a simple calculation that $\frac{\partial OFP_{1}}{\partial
a}(a,b)=\frac{\mu}{\sqrt{b}\sigma} \cdot \phi(\Delta_1) > 0$.
Similarly, if $OFP_{2}(a,b)$ denotes the overflow probability in the
second bin when the first bin has total mean $a\mu$ and total
variance $bV$, then  $\frac{\partial OFP_{2}}{\partial
a}=\frac{-\mu}{\sqrt{1-b}\sigma} \cdot \phi(\Delta_2) < 0$. Thus,
$OFP_{1}$ is monotonically increasing in $a$ and $OFP_{2}$ is
monotonically decreasing in $a$, and therefore there is a unique
minimum for $OFP(a,b)$ (when $b$ is fixed and $a$ is free) that is
obtained when $OFP_{1}(a,b)=OFP_{2}(a,b)$, i.e., when
$\Delta_1=\Delta_2$.

\item [item (\ref{it:saddle}):] We first explicitly determine what
$WOFP$ restricted to the valley is as a function $D(b)=WOFP(m(b),b)$
of $b$. From before we know that on the valley $\Delta_1=\Delta$.
Therefore, following the same reasoning as in the SP-MED case,
\begin{eqnarray*}
D(b) &=& \frac{c-\mu}{\sigma}\frac{1}{\sqrt{b}+\sqrt{1-b}}.
\end{eqnarray*}
It follows that $D(b)$ is monotonically decreasing in $b$ for $b \le
\frac{1}{2}$ and increasing otherwise. The maximal point is obtained
in the saddle point that is the center of the symmetry.
\end{description}
\end{proof}

%\red{Remove? Figure \ref{fig:normal_ofp_a_b_dots_3d_curve} depicts
%$WOFP_{c_1,c_2}(a,b)$ for the case of $c_1=c_2$. It also shows the
%bottom sorted line as well as the optimal fractional solution.}

%%\red{Later}
%\begin{figure}
%  % Requires \usepackage{graphicx}
%  \begin{center}
%  \includegraphics[width=6cm]{normal_ofp_a_b_dots_3d_curve_png}\\
%  \caption{$OFP(a,b)$. $C=590, \mu=492, V=84977, c_1=236, c_2=354, n=10$. The bottom sorted line is in black. The purple dot represents the optimal fractional point.}
%  \label{fig:normal_ofp_a_b_dots_3d_curve}
%  \end{center}
%\end{figure}

\begin{lemma}
\label{lem:spmwop:error} The difference between minimal worst
overflow probability in the integral point found by the sorting
algorithm and the optimal integral (or fractional) point for SP-MWOP
is at most $O(\frac{L}{\alpha n})$. In particular, when
$L=\littleo{n}$ and $\alpha$ is a constant, the difference is
$\littleo{1}$.
\end{lemma}

\begin{proof}
We know from Theorem \ref{thm:error} that the difference is at most
$$
\min \set{|\nabla D(\xi_1)| , |\nabla D(\xi_2)|} \frac{L}{n},$$
where $\xi_1=(a_1,b_1) \in [O_1,OPT_f]$, $\xi_2=(a_2,b_2) \in
[OPT_f,O_2]$ and $O_1$ and $O_2$ are the two points on the bottom
sorted path between which $OPT_f$ lies, and notice that even though
$WOFP$ is not differentiable when $\Delta_1=\Delta_2$, it is
differentiable everywhere else. We now use the partial derivatives
calculated in Appendix \ref{app:spmed:partial_derivatives}.
%
%$$
%\min \set{|(-(\frac{\phi(\Delta_2) \mu}{\sigma_2})(\xi_1),
%-(\frac{\Delta_2
%\phi(\Delta_2)}{2(1-b)})(\xi_1))|,|((\frac{\phi(\Delta_1)\mu}{\sigma_1})(\xi_2),(\frac{\Delta_1
%\phi(\Delta_1)}{2b})(\xi_2)|)} \frac{L}{n}.$$
%
We also replace $\frac{\phi(\Delta_2)}{\sigma_2}$ with
$\frac{\Delta_2 \phi(\Delta_2)}{c_2-(1-a)\mu}$ and similarly for the
other term. We get:

$$
\min \set{|\Delta_2 \phi(\Delta_2)| \cdot
|(\frac{\mu}{c_2-(1-a_1)\mu}, \frac{1}{2(1-b_1)})|,|\Delta_1
\phi(\Delta_1)| \cdot |(\frac{\mu}{c_1-a_2\mu},\frac{1}{2b_2})|}
\frac{L}{n}.$$

$\Delta \phi(\Delta)$ maximizes at $\Delta=1$ with value at most
$\frac{1}{\sqrt{2 \pi e}}$. Also, $(c_1 - a_2\mu) + (c_2 -
(1-a_1)\mu) = c - \mu -(a_2-a_1)\mu \ge c-\mu \frac{L}{n} \ge
\frac{\alpha}{2}\mu$, where $\alpha$ is the total space capacity,
and a constant by our assumption. Hence, at least one of the terms
$\frac{\mu}{c_2-(1-a_1)\mu}$,$\frac{\mu}{c_1-a_2\mu}$ is at most
$\frac{4}{\alpha}$. Also, for that term, the spare capacity is
maximal, and therefore it takes at least half of the variance.
Altogether, the difference is at most $O(\frac{L}{\alpha n})$ which
completes the proof.
\end{proof}

\subsection{Stochastic Packing with Minimum Overflow Probability (SP-MOP)}

The SP-MOP problem gets as input integers $k$ and $n$, specifying
the number of bins and services, integers $c_1, \ldots , c_k$,
specifying the bin capacities and values
$\set{(\mu^\ui,V^\ui)}_{i=1}^n$, specifying that the demand
distribution $X^\ui$ of service $i$ is normal with mean $\mu^\ui$
and variance $V^\ui$. A solution to the problem is a partition of
$[n]$ to $k$ disjoint sets $S_1,\ldots,S_k \subseteq [n]$ that
minimizes the overflow probability.

The total overflow probability is $OFP=1-\prod_{j=1}^k (1-OFP_j)$
where as we computed before (in Appendix
\ref{app:spmwop:single_bin}) $OFP_j=1-\Phi(\Delta_j)$. With two bins
$OFP$ is a function from $[0,1]^2$ to $\R$ and
$OFP(a,b)=1-\Phi(\Delta_1)\Phi(\Delta_2)$ where the first bin has
mean $a\mu$ and variance $bV$, the second bin has mean $(1-a)\mu$
and variance $(1-b)V$ and $\sigma_j,\Delta_j$ are defined as above.

\begin{lemma}
$OFP$ respects conditions (\ref{it:kto2})-(\ref{it:valley}) of
Section \ref{sec:cost}.
\end{lemma}

\begin{proof}
We go over the conditions one by one.
\begin{description}
\item [item (\ref{it:kto2}):] Let $S_1,\ldots,S_k$ be an optimal solution. Suppose there are
$j$ and $j'$ for which $S_j,S_{j'}$ are not the optimal solution for
the two-bin problem defined by the services in $S_j \cup S_{j'}$ and
capacities $c_j, c_{j'}$. Change the allocation of the solution
$S_1,\ldots,S_k$ on bins $j$ and $j'$ to an optimal solution for the
two bin problem. This change improves $(1-OFP_j)(1-OFP_{j'})$ while
not affecting the overflow probability of any other bin. In total we
get a better solution, in contradiction to the optimality of
$S_1,\ldots,S_k$.

%\item [item (\ref{it:diffrentiable}):] Clearly $OFP$ is **not**
%differentiable.

\item [item (\ref{it:symmetry}):] The same proof as in Appendix
\ref{app:spmed:symmetry} shows
$OFP(a,b)=OFP(1-a-\frac{c_2-c_1}{\mu}, 1-b)$.

\item [item (\ref{it:valley}):]
Fix $b$.
\begin{eqnarray*}
\frac{\partial^2 OFP}{\partial^2 a} & = & \mu^2 [
\frac{\Delta_1}{\sigma_1^2} \phi(\Delta_1) \Phi(\Delta_2) +
\frac{\Delta_2}{\sigma_2^2} \phi(\Delta_2) \Phi(\Delta_1) +
\frac{2}{\sigma_1 \sigma_2} \phi(\Delta_1)\phi(\Delta_2)].
\end{eqnarray*}
In particular $\frac{\partial^2 OFP}{\partial^2 a} >0$ and for every
fixed $b$, $OFP(a,b)$ is convex over $a \in [0..1]$ and has a unique
minimum $a=m(b)$.
\end{description}
\end{proof}.

Proving there exists a unique maximum over the valley is more
challenging. We wish to find all extremum points of the cost
function $D$ ($OFP$ in our case) over the valley $\set{(m(b),b)}$.
Define $m(a,b)=a-m(b)$. Then we wish to maximize $D(a,b)$ subject to
$m(a,b)=0$. Before, we computed the restriction $D(b)$ of the cost
function over the valley and found its extremum points. However,
here we do not know how to explicitly find $D(b)$. Instead, we use
Lagrange multipliers that allow working with the implicit form
$m(a,b)=0$ without explicitly finding $D(b)$. We prove a general
result:

\begin{lemma}
\label{lem:lagrange} If a cost function $D$ is differentiable twice
over $[0,1] \times [0,1]$ and respects conditions
(\ref{it:kto2})-(\ref{it:saddle}) of Section \ref{sec:cost}, and if
for every $b \in [0,1]$, $D(a,b)$ is strictly convex for $a \in
[0,1]$, then any extremum point of $D$ over the valley must have
zero gradient at $Q$, i.e., $\nabla(D)(Q)=0$.
% $\frac{\partial^2 OFP}{\partial^2 a} >0$
\end{lemma}

\begin{proof}
The valley is defined by the equation $m(a,b)=\frac{\partial
D}{\partial a}(a,b)$. Using Lagrange multipliers we find that at any
extremum point $Q$ of $D$ over the valley,
$$\nabla(D)(Q)=\lambda \nabla m(Q).$$
For some real value $\lambda$. However,
\begin{eqnarray*}
\nabla(D)(Q) &=& (\frac{\partial D}{\partial a}(Q),\frac{\partial
D}{\partial b}(Q)) ~=~ (0,\frac{\partial D}{\partial b}(Q)),
\end{eqnarray*}
because $Q$ is on the valley. H

Also, as for every fixed $b$, $D(a,b)$ is strictly convex, we have
that \begin{eqnarray*} \frac{\partial m}{\partial a}(Q) &=&
\frac{\partial^2 D}{\partial^2 a}(Q) ~>~ 0. \end{eqnarray*}
Therefore, we conclude that $\lambda=0$. This implies that
$\frac{\partial D}{\partial b}(Q)=0$. Hence, $\nabla(D)(Q)=0$.
\end{proof}

With that we prove:

\begin{lemma}
$OFP$ respects condition (\ref{it:saddle}) of Section
\ref{sec:cost}.
\end{lemma}

\begin{proof}
Let $Q=(a,b)$ be an extremum point of $OFP$ over the valley. We look
at the range $b \in [0..\half)$, $b \ge \half$ is obtained by the
symmetry. Then, by Lemma \ref{lem:lagrange}:

\begin{eqnarray*}
\phi(\Delta_1)\Phi(\Delta_2) \frac{\partial \Delta_1}{\partial a}
&=& -\phi(\Delta_2)\Phi(\Delta_1) \frac{\partial \Delta_2}{\partial
a} \mbox{               , and}\\
\phi(\Delta_1)\Phi(\Delta_2) \frac{\partial \Delta_1}{\partial b}
&=& -\phi(\Delta_2)\Phi(\Delta_1) \frac{\partial \Delta_2}{\partial
b}.
\end{eqnarray*}

Dividing the two equations we get

\begin{eqnarray*}
\frac{\partial \Delta_1}{\partial a} \frac{\partial
\Delta_2}{\partial b} &=& \frac{\partial \Delta_2}{\partial a}
\frac{\partial \Delta_1}{\partial b}.
\end{eqnarray*}

Plugging the partial derivatives of $\Delta_i$ by $a$ and $b$, we
get the equation

\begin{eqnarray*}
\frac{\Delta_1}{\Delta_2} & = & \sqrt{\frac{b}{1-b}}.
\end{eqnarray*}

As $b \le \half$, $b < 1-b$ and we conclude that at $Q$ $\Delta_1 <
\Delta_2$. However, using the log-concavity of the normal c.d.f
function $\Phi$ we prove in Appendix \ref{sec:spmop} that:

\begin{lemma}
\label{lem:OFP}
$\frac{\partial OFP}{\partial a}=0$ at a point
$Q=(a,b)$ with $b \le \half$ implies $\Delta_1 \ge \Delta_2$.
\end{lemma}

Together, this implies that the only extremum point of $OFP$ over
the valley is at $b=\half$. However, at $b=0$, the best is to fill
the largest bin to full capacity with variance $0$, and thus,
$OFP(m(0),0)=1-\Phi(\Delta)$ where $\Delta=\frac{c-\mu}{\sigma}$. On
the other hand, at $b=\half$, $OFP(a=m(\half),\half)=
1-\Phi(\frac{c_1-a
\mu}{\sqrt{\half}\sigma})\Phi(\frac{c_2-(1-a)\mu}{\sqrt{\half}\sigma})$.
As $(c_1-a \mu)+(c_2-(1-a) \mu)=c-\mu$, either $c_1-a \mu$ or
$c_2-(1-a) \mu$ is at most $\frac{c-\mu}{2}$ and therefore
$\Phi(\sqrt{2} \frac{c_1-a
\mu}{\sigma})\Phi(\sqrt{2}\frac{c_2-(1-a)\mu}{\sigma}) \le
\Phi(\sqrt{2} \frac{c-\mu}{2\sigma})=\Phi(\frac{c-\mu}{\sqrt{2}
\sigma}) \le \Phi(\frac{c-\mu}{\sigma})$. We conclude that
$OFP(a,\half) \ge OFP(m(0),0)$ and there is a unique maximum point
on the valley and it is obtained at $b=\half$.
\end{proof}

\section{Bounding the approximation error of the sorting algorithm}
\label{sec:error}

We assume that no input service is too \emph{dominant}. Recall that
we represent service $i$ with the point $P^\ui=(a^\ui,b^\ui)$ and
$P^\uone+P^\utwo+\ldots+P^\un=(1,1)$. Thus, $\sum_i |P^\ui| \ge
|(1,1)| = \sqrt{2}$ (by the triangle inequality) and  $\sum_i
|P^\ui| \le 2$ (because the length of the longest increasing path
from $(0,0)$ to $(1,1)$, is obtained by the path going from $(0,0)$
to $(1,0)$ and then to $(1,1)$). Hence, the average length of an
input point $P^\ui$ is somewhere between $\frac{\sqrt{2}}{n}$ and
$\frac{2}{n}$. Our assumption states that no element takes more than
$L$ times its "fair" share, i.e., that for every $i$, $|P^\ui| \le
\frac{L}{n}$. With that we prove:

\begin{theorem}
\label{thm:error} Let $OPT_f$ be the fractional optimal solution. If $D$ is differentiable, the difference between
the cost on the integral point found by the sorting
algorithm and the cost on the optimal integral (or fractional) point is at most
$\min \set{|\nabla D(\xi_1)| , |\nabla D(\xi_2)|} \frac{L}{n}$,
where $\xi_1 \in [O_1,OPT_f]$, $\xi_2 \in [OPT_f,O_2]$ and $O_1$ and
$O_2$ are the two points on the bottom sorted path between which
$OPT_f$ lies.
\end{theorem}

%This shows the approximation factor goes to $1$ and linearly (in the number of services) fast. Thus, from a practical point of view the theorem is very satisfying.

\begin{proof}
Suppose we run the sorting algorithm on some input. Let $OPT_{int}$ be the integral optimal solution, $OPT_f$ the fractional optimal solution and $OPT_{sort}$ the integral point the sorting algorithm finds on the bottom sorted path. We wish to bound $D(OPT_{sort})-D(OPT_{int})$ and clearly it is at most $D(OPT_{sort})-D(OPT_f)$. We now look at the two points $O_1$ and $O_2$ on the bottom sorted path between which $OPT_f$ lies (and notice that as far as we know it is possible that $OPT_{sort}$ is none of these points). Since $D(OPT_f) \le D(OPT_{sort}) \le D(O_1)$ and $D(OPT_f) \le D(OPT_{sort}) \le D(O_2)$ the error the sorting algorithm makes is at most $$\min \set{D(O_1)-D(OPT_f),D(O_2)-D(OPT_f)}.$$ We now apply the mean value theorem and use our assumption that for every $i$, $|P^\ui| \le \frac{L}{n}$.
\end{proof}

We define a new system constant, \emph{relative spare capacity},
denoted by $\alpha$ which equals $\frac{c -\mu }{\mu}$, i.e., it
expresses the spare capacity as a fraction of the total mean. We
assume that the system has some constant (possibly small) relative
spare capacity. Also, we only consider solutions where each bin is
allocated services with total mean not exceeding its capacity.
Equivalently, we only consider solutions where $\Delta_j \ge 0$ for
every $1 \le j \le k$. We will later see that under these
conditions the sorting algorithm solves all three cost functions
with a small error going fast to zero with $n$.

We remark that in fact the proof shows something stronger: the
deviation of any (not necessarily optimal) fractional solution on
the bottom sorted path, is close to the deviation of the integral
solution to the left or to the right of it on the bottom sorted
path.

\section{Proof of Lemma \ref{lem:epigraph_integral}}
\label{app:epigraph_integral_proof}
\begin{proof}
We introduce some notation. Let $\tau=\tau_1,\ldots,\tau_n$ be a sequence of $n$ elements that is a reordering of $\set{1,\ldots,n}$. We associate with $\tau$ the $n$ partial sums
$P_\tau^{[1]},\ldots,P_\tau^{[n]}$ where $P_\tau^{[i]}$ is $\sum_{j=1}^i P^{(\tau_j)}$, i.e., $P_\tau^{[i]}$ is the integral point that is the sum of the first $i$ points according to the sequence $\tau$. We also define $P_\tau^{[0]}=(0,0)$ and $P_\tau^{[n]}=(1,1)$. The \emph{curve connecting $\tau$} is the curve that is formed by connecting $P_\tau^{[i]}$ and $P_\tau^{[i+1]}$ with a line, for $i=0,\ldots,n-1$.

Assume that in the sequence $\tau=\tau_1,\ldots,\tau_n$ there is some index $i$ such that the VMR of $P^{(\tau_i)}$ is larger than the VMR of $P^{(\tau_{i+1})}$. Consider the sequence $\tau'$ that is the same as $\tau$ except for switching the order of $\tau_i$ and $\tau_{i+1}$. I.e., $\tau'=\tau_1,\ldots,\tau_{i-1},\tau_{i+1},\tau_i,\tau_{i+2},\ldots,\tau_n$. We claim that the curve connecting $\tau'$ lies beneath the curve connecting $\tau$. To see that notice that both curves are the same up to the point $P_\tau^{[i-1]}$. There, the two paths split. $\tau$ adds $P^{(\tau_{i})}$ and then $P^{({\tau_{i+1}})}$ while $\tau'$ first adds $P^{(\tau_{i+1})}$ and then $P^{({\tau_{i}})}$. Then the two curves coincide and overlap all the way to $(1,1)$. In the section where the two paths differ, the two different paths form a parallelogram with $P^{(\tau_i)}$ and $P^{(\tau_{i+1})}$ as two neighboring edges of the parallelogram. As the angle $P^{(\tau_{i+1})}$ has with the $a$ axis is smaller than the angle $P^{(\tau_{i})}$ has with the $a$ axis, the curve connecting $\tau'$ goes beneath that of $\tau$.

To finish the argument, let $P_I$ be an arbitrary integral point for some $I \subseteq [n]$. Look at the sequence $\tau$ that starts with the elements of $I$ followed by the elements of $[n] \setminus I$ in an arbitrary order. Notice that $P_I$ lies on the curve connecting $\tau$. Now run a bubble sort on $\tau$, each time ordering a pair of elements by their VMR. Notice that the process terminates with the sequence that sorts the elements by their VMR and the curve connecting the final sequence is the bottom sorted path. Thus, we see that the bottom sorted path lies beneath the curve connecting $\tau$, and in particular $P_I$ lies above the bottom sorted path. A similar argument shows $P_I$ lies underneath the upper sorted path.
\end{proof}

\section{Proof of Theorem \ref{thm:fractional}: The optimal solution is on the bottom sorted path} 
\label{app:fractional_proof}

\begin{proof}
Consider an arbitrary fractional point $(a_0,b_0)$ lying strictly inside the polygon confined by the upper and bottom sorted paths. If $b_0 \le \half$, then by keeping $b=b_0$ constant and changing $a$ till it reaches the valley we strictly decrease cost (because $D$ is strictly monotone in this range). Now, when changing $a$ we either hit the bottom sorted path or the valley. If we hit the bottom sorted path, we found a point on the bottom sorted path with less cost and we are done. If we hit the valley, we can go down the valley until we hit the bottom sorted path and again we are done (as $D$ is strictly monotone on the valley).

If $b_0 \ge \half$ we recall two facts that we already know:
\begin{itemize}
\item The point $(1-a_0,1-b_0)=\varphi(a_0,b_0)$ is fractional (since $(a_0,b_0)$ is fractional and $\varphi$ maps fractional points to fractional points), and,
\item By the reflection symmetry we know that $D(a_0,b_0)=D(1-a_0-\zeta,1-b_0)$ where $\zeta=\frac{c_2-c_1}{\mu} \ge 0$.
\end{itemize}
Now, $(1-a_0-\zeta,1-b_0)$ has $b$ coordinate that is at most $\half$. Also  $(1-a_0-\zeta,1-b_0)$ lies to the \emph{left} of the fractional point $(1-a_0,1-b_0)$ (since $\zeta>0$) and therefore it lies above the bottom sorted path. We therefore see that the point $(a_0,b_0)$ has a corresponding fractional point with the same cost and with $b$ coordinate at most $\half$. Applying the argument that appears in the first paragraph of the proof we conclude that there exists some point on the bottom sorted path with less cost, and conclude the proof.
\end{proof}

\section{Standard Normal Distribution}
\label{app:standard_normal_distribution}

The probability density function of the standard normal distribution (that has mean $0$ and variance $1$) is $\phi(x)= \frac{1}{\sqrt{2\pi}} e^{-\frac{x^2}{2}}$ and the cumulative distribution function is $\Phi(x)= \frac{1}{\sqrt{2\pi}} \int_{-\infty}^x e^{-\frac{t^2}{2}}dt$.
Clearly, $\Phi'(x)= \phi(x)$. Also, $\phi'(x)=  -\frac{x}{\sqrt{2\pi}} e^{-\frac{x^2}{2}} =
-x \phi(x)$. The second derivative is
$\phi''(x)=
- \phi(x) + x^2 \phi(x) =
(x^2 - 1) \phi(x)$.

\section{SP-MED}
\subsection{Expected Deviation of a single bin}
\label{app:spmed:dev_single_bin}

By definition the expected deviation of a single bin is $Dev_{S_j}=
\frac{1}{\sigma_j \sqrt{2 \pi}} \int_{c_j}^{\infty} (x - c_j)
e^{-\frac{(x-\mu_j)^2}{2\sigma_j^2}} dx$. Doing the variable change
$t = \frac{x - \mu_j}{\sigma_j}$
%
%we get:
%
%\begin{eqnarray*}
%Dev_{S_j}
%& = &
%\frac{1}{\sqrt{2 \pi}} \int_{\frac{c_j - \mu_j}{\sigma_j}}^{\infty} (t \sigma_j + \mu_j - c_j) e^{-\frac{t^2}{2}} dt \\
%& = &
%\frac{1}{\sqrt{2 \pi}} \int_{\frac{c_j - \mu_j}{\sigma_j}}^{\infty} (\mu_j - c_j) e^{-\frac{t^2}{2}} dt + \frac{\sigma_j}{\sqrt{2 \pi}} \int_{\frac{c_j - \mu_j}{\sigma_j}}^{\infty} t e^{-\frac{t^2}{2}} dt
%\end{eqnarray*}
%
%Doing
and then the variable change  $y = \frac{-t^2}{2}$ we get:

\begin{eqnarray*}
Dev_{S_j} & = & (\mu_j - c_j) [1 - \Phi(\frac{c_j -
\mu_j}{\sigma_j})] -
 \frac{\sigma_j}{\sqrt{2 \pi}} \int_{- \frac{1}{2} (\frac{c_j - \mu_j}{\sigma_j})^2}^{- \infty} e^{y} dy \\
%& = &
%-\sigma_j \Delta_j [1 - \Phi(\Delta_j)] - \frac{\sigma_j}{\sqrt{2 \pi}} e^y |_{- \frac{1}{2} \Delta_j^2}^{- \infty} \\
& = & -\sigma_j \Delta_j [1 - \Phi(\Delta_j)] +
\frac{\sigma_j}{\sqrt{2 \pi}} e^{- \frac{1}{2} \Delta_j^2}  =
\sigma_j [ \phi(\Delta_j)-\Delta_j (1-\Phi(\Delta_j))].
\end{eqnarray*}

Denoting  $g(\Delta)= \phi(\Delta) - \Delta [1-\Phi(\Delta)]$ we see that $Dev_{S_j}=\sigma_j g(\Delta_j)$.

\subsection{Symmetry} \label{app:spmed:symmetry}

\begin{claim}
\label{clm:symmetry} $Dev(a,b)=Dev(1-a-\frac{c_2-c_1}{\mu}, 1-b)$.
\end{claim}

\begin{proof}
Let us define $\sigma_1(b)=\sqrt{b}~\sigma$,
$\sigma_2(b)=\sqrt{1-b}~\sigma$,
$\Delta_1(a,b)=\frac{c_1-a\mu}{\sigma_1(b)}$ and
$\Delta_2(a,b)=\frac{c_2-(1-a)\mu}{\sigma_2(b)}$. We know that
$Dev(a,b)=\sigma_1(b) g(\Delta_1(a,b))+\sigma_2(b)
g(\Delta_2(a,b))$. To prove the claim it is enough to show that the
following four equations hold: $\sigma_1(b)=\sigma_2(1-b)$,
$\sigma_2(b)=\sigma_1(1-b)$,
$\Delta_1(a,b)=\Delta_2(1-a+\frac{c_1-c_2}{\mu}, 1-b)$ and
$\Delta_2(a,b)=\Delta_1(1-a+\frac{c_1-c_2}{\mu}, 1-b)$.

Indeed, $\sigma_1(1-b) = \sqrt{1-b} ~\sigma = \sigma_2(b)$ and
similarly $\sigma_2(1-b) = \sigma_1(b)$.
%\begin{eqnarray*}
%\sigma_2(1-b) = \sqrt{1-(1-b)}~\sigma = \sqrt{b}~\sigma = \sigma_1(b).
%\end{eqnarray*}
Also,
\begin{eqnarray*}
\Delta_2(1-a-\frac{c_2-c_1}{\mu}, 1-b) & = &
\frac{c_2-(1-(1-a+\frac{c_1-c_2}{\mu}))\mu}{\sigma_2(1-b)} \\
%& = &
%\frac{c_2-a\mu+c_1-c_2}{\sigma_1(b)} =  \frac{c_1-a\mu}{\sigma_1(b)} = \Delta_1(a,b).
\end{eqnarray*}
A similar check shows that $\Delta_1(1-a-\frac{c_2-c_1}{\mu}, 1-b) =
\Delta_2(a,b)$.
%\begin{eqnarray*}
%\Delta_1(1-a+\frac{c_1-c_2}{\mu}, 1-b) = \frac{c_1-(1-a+\frac{c_1-c_2}{\mu})\mu}{\sigma_1(1-b)} = \frac{c_1-(1-a)\mu-c_1+c_2}{\sigma_2(b)} = \frac{c_2-(1-a)\mu}{\sigma_2(b)} = \Delta_2(a,b).
%\end{eqnarray*}
\end{proof}

\subsection{The partial derivatives of $Dev$}
\label{app:spmed:partial_derivatives}

The system has absolute constants $\mu$ and $V$, $\sigma=\sqrt{V}$. We let $\sigma_1=\sqrt{b}\sigma$ and $\sigma_2=\sqrt{1-b}\sigma$.
We let $\Delta_1= \frac{c_1 - \mu_1}{\sigma_1} = \frac{c_1 - a \mu}{\sigma_1}$ and  $\Delta_2= \frac{c_2 - \mu_2}{\sigma_2} = \frac{c_2 - (1-a) \mu}{\sigma_2}$. In this notation, $Dev(a,b)=\sigma_1~ g(\Delta_1) + \sigma_2~ g(\Delta_2)$.

We calculate the $g$ derivatives. Since $g(\Delta)= \phi(\Delta) +
\Delta \Phi(\Delta) - \Delta$ we have that $g'(\Delta)=
\phi'(\Delta)+ \Phi(\Delta) + \Delta \phi(\Delta) - 1= -\Delta
\phi(\Delta) + \Phi(\Delta) + \Delta \phi(\Delta) - 1 =
-(1-\Phi(\Delta))$ and $g''(\Delta)= \phi(\Delta)$. As $\Phi(\Delta)
\le 1$, $g'(\Delta) \le 0$ and $g$ is monotonically decreasing. As
$g''(\Delta)>0$, $g$ is convex.

\subsubsection{Deriving by $a$} \label{app:spmed:partial_a_derivative}

We have:
$\frac{\partial \Delta_1}{\partial a}
=
-\frac{\mu}{\sigma_1}$ and
$\frac{\partial \Delta_2}{\partial a}
=
\frac{\mu}{\sigma_2}$. Also,

\begin{eqnarray*}
\frac{\partial Dev}{\partial a}(a,b)
& = &
\sigma_1 \frac{\partial g(\Delta_1)}{\partial \Delta_1}~ \frac{\partial \Delta_1}{\partial a} + \sigma_2 \frac{\partial g(\Delta_2)}{\partial \Delta_2}~ \frac{\partial \Delta_2}{\partial a} \\
& = & - \sigma_1 \frac{\mu}{\sigma_1}~ [\Phi(\Delta_1) - 1] +
\sigma_2 \frac{\mu}{\sigma_2}~ [\Phi(\Delta_2) - 1]   =  \mu~
[\Phi(\Delta_2) - \Phi(\Delta_1)]
\end{eqnarray*}

Since $\Phi$ is monotonic, a zero value is achieved when $\Delta_1 = \Delta_2$. Deriving again by $a$:

\begin{eqnarray*}
\frac{\partial^2 Dev}{\partial a^2} & = & \mu~ [~\frac{\partial
\Phi(\Delta_2)}{\partial \Delta_2}~ \frac{\partial
\Delta_2}{\partial a} - \frac{\partial \Phi(\Delta_1)}{\partial
\Delta_1}~ \frac{\partial \Delta_1}{\partial a}~]   =  \mu^2~
[\frac{\phi(\Delta_2)}{\sigma_2} + \frac{\phi(\Delta_1)}{\sigma_1}]
\end{eqnarray*}

Since $\frac{\partial^2 Dev}{\partial a^2} \geq 0$, it follows that for any $0 < b < 1$, $Dev(a)$ is convex and hence $\Delta_1 = \Delta_2$ is a minimum point in this range.

\subsubsection{Deriving by $b$} \label{app:spmed:partial_b_derivative}

We have: $\frac{\partial \sigma_1}{\partial b}=
\frac{\sigma}{2\sqrt{b}}$ and
$\frac{\partial \sigma_2}{\partial b}
=
-\frac{\sigma}{2\sqrt{1-b}}$. Also,
$\frac{\partial \Delta_1}{\partial b}
=
-\frac{\Delta_1}{2b}$ and
$\frac{\partial \Delta_2}{\partial b}
=
\frac{\Delta_2}{2(1-b)}$. Now,

\begin{eqnarray*}
\frac{\partial Dev}{\partial b} & = & \frac{\partial
\sigma_1}{\partial b} g(\Delta_1) + \sigma_1 \frac{\partial
g(\Delta_1)}{\partial \Delta_1}~ \frac{\partial \Delta_1}{\partial
b} +
\frac{\partial \sigma_2}{\partial b} g(\Delta_2) + \sigma_2 \frac{\partial g(\Delta_2)}{\partial \Delta_2}~ \frac{\partial \Delta_2}{\partial b} \\
%& = &
%\frac{\sigma}{2\sqrt{b}}~ [\phi(\Delta_1) + \Delta_1 \Phi(\Delta_1) - \Delta_1] -\sigma \sqrt{b}~ [\Phi(\Delta_1) - 1] \frac{\Delta_1}{2b} \\
%& &
%-\frac{\sigma}{2\sqrt{1-b}}~ [\phi(\Delta_2) + \Delta_2 \Phi(\Delta_2) - \Delta_2] +\sigma \sqrt{1-b}~ [\Phi(\Delta_2) - 1] \frac{\Delta_2}{2(1-b)} \\
& = &
\frac{\sigma}{2} [\frac{\phi(\Delta_1)}{\sqrt{b}} - \frac{\phi(\Delta_2)}{\sqrt{1-b}}]  ~ = ~
\frac{V}{2} [\frac{\phi(\Delta_1)}{\sigma_1} - \frac{\phi(\Delta_2)}{\sigma_2}]
\end{eqnarray*}

Deriving a second time we get:

\begin{eqnarray*}
\frac{\partial^2 Dev}{\partial b^2} & = & \frac{\sigma}{2}
[\frac{1}{\sqrt{b}}~ \frac{\partial \phi(\Delta_1)}{\partial
\Delta_1}~ \frac{\partial \Delta_1}{\partial b} -
\frac{1}{2b\sqrt{b}}~ \phi(\Delta_1) -
\frac{1}{\sqrt{1-b}}~ \frac{\partial \phi(\Delta_2)}{\partial \Delta_2}~ \frac{\partial \Delta_2}{\partial b} - \frac{1}{2(1-b)\sqrt{1-b}}~ \phi(\Delta_2)~] \\
%& = &
%\frac{\sigma}{2} [\frac{1}{\sqrt{b}}~ \Delta_1~ \phi(\Delta_1)~ \frac{\Delta_1}{2b}~ - \frac{1}{2b\sqrt{b}}~ \phi(\Delta_1) \\
%& &
%+ \frac{1}{\sqrt{1-b}}~ \Delta_2~ \phi(\Delta_2)~ \frac{\Delta_2}{2(1-b)}~ - \frac{1}{2(1-b)\sqrt{1-b}}~ \phi(\Delta_2)~] \\
& = &
\frac{\sigma}{4} [\frac{\phi(\Delta_1)}{b\sqrt{b}}~ (\Delta_1^2 - 1) + \frac{\phi(\Delta_2)}{(1-b)\sqrt{1-b}}~ (\Delta_2^2 - 1) ] \\
\end{eqnarray*}

\subsubsection{The Mixed Derivative}

\begin{eqnarray*}
\frac{\partial^2 Dev}{\partial b \partial a} & = &
\frac{\partial}{\partial b} \frac{\partial Dev}{\partial a}   =
\frac{\partial}{\partial b} [ ~\mu [\Phi(\Delta_2) -
\Phi(\Delta_1)]~ ]   =
\mu~ [~\frac{\partial \Phi(\Delta_2)}{\partial \Delta_2}~ \frac{\partial \Delta_2}{\partial b} - \frac{\partial \Phi(\Delta_1)}{\partial \Delta_1}~ \frac{\partial \Delta_1}{\partial b}~] \\
& = & \mu~ [\phi(\Delta_2) \frac{\Delta_2}{2 (1-b)} + \phi(\Delta_1)
\frac{\Delta_1}{2b}]   =  \mu~ [\frac{1}{2b} \Delta_1 \phi(\Delta_1)
+ \frac{1}{2 (1-b)} \Delta_2 \phi(\Delta_2)]
\end{eqnarray*}

\subsubsection{The saddle point} \label{app:spmed:saddle}

The Hessian of the function $Dev$ contains the second order partial derivatives of $Dev$, i.e., it is a $2 \times 2$ matrix $H$,

\begin{eqnarray*}
H(a,b) & = & \left( \begin{array}{cc}
                 \frac{\partial^2 Dev}{\partial a^2}(a,b) & \frac{\partial^2 Dev}{\partial b \partial a}(a,b) \\
                \frac{\partial^2 Dev}{\partial a \partial b}(a,b) & \frac{\partial^2 Dev}{\partial b^2}(a,b)
               \end{array}
\right)
\end{eqnarray*}
$H$ is symmetric and therefore at any point $(a,b)$ the matrix $H(a,b)$ has two real eigenvalues. Now, assume the first order derivatives of $Dev$ vanish at some point $(a_0,b_0)$. If the Hessian at that point $(a_0,b_0)$ has negative determinant, then the product of these two eigenvalues is negative, and hence one is positive and the other negative. It then follows that the point $(a_0,b_0)$ is a \emph{saddle point} of $Dev$.

We now compute the Hessian at the point $(a_0, b_0) = (\half-\frac{c_2-c_1}{2D}, \frac{1}{2})$ which is the center of the reflection symmetry. We know that at this point $\Delta_1 = \Delta_2$ and $\sigma_1 = \sigma_2$. If then follows that:

\begin{eqnarray*}
H(a_0,b_0) & = & \left( \begin{array}{cc}
                 2 \sqrt{2}~ \mu^2~ \frac{\phi(\Delta_1)}{\sigma} & 2~ \mu \Delta_1 \phi(\Delta_1) \\
                2~ \mu \Delta_1 \phi(\Delta_1) & \sqrt{2}~ \sigma~ \phi(\Delta_1)~ [\Delta_1^2~ - 1]
               \end{array}
\right)
\end{eqnarray*}

%\begin{eqnarray*}
%\frac{\partial^2 Dev}{\partial a^2}(a_0, b_0)
%& = &
%\mu^2~ [\frac{\phi(\Delta_1)}{\sigma_1} + \frac{\phi(\Delta_1)}{\sigma_1}] \\
%& = &
%2 \mu^2~ \frac{\phi(\Delta_1)}{\sigma_1} \\
%& = &
%2 \sqrt{2}~ \mu^2~ \frac{\phi(\Delta_1)}{\sigma}
%\end{eqnarray*}

%\begin{eqnarray*}
%\frac{\partial^2 Dev}{\partial b^2}(a_0, b_0)
%& = &
%\frac{\sigma}{4} [\frac{\phi(\Delta_1)}{b\sqrt{b}}~ (\Delta_1^2 - 1) + \frac{\phi(\Delta_1)}{(1-b)\sqrt{1-b}}~ (\Delta_1^2 - 1) ] \\
%& = &
%\sqrt{2}~ \sigma~ \phi(\Delta_1)~ [\Delta_1^2~ - 1]
%\end{eqnarray*}

%\begin{eqnarray*}
%\frac{\partial^2 Dev}{\partial a \partial b}(a_0, b_0)
%& = &
%\mu~ [\frac{1}{2b} \Delta_1 \phi(\Delta_1) + \frac{1}{2 (1-b)} \Delta_1 \phi(\Delta_1)] \\
%& = &
%2~ \mu \Delta_1 \phi(\Delta_1)
%\end{eqnarray*}

We can now compute the determinant:

\begin{eqnarray*}
\det[H(a_0, b_0)] & = & \frac{\partial^2 Dev}{\partial a^2}(a_0,
b_0)~ \frac{\partial^2 Dev}{\partial b^2}(a_0, b_0) -
  [\frac{\partial^2 Dev}{\partial a \partial b}(a_0, b_0)]^2 \\
& = &
2 \sqrt{2}~ \mu^2~ \frac{\phi(\Delta_1)}{\sigma}~ \sqrt{2}~ \sigma~ \phi(\Delta_1)~ [\Delta_1^2~ - 1] -  4 \mu^2 \Delta_1^2~ [\phi(\Delta_1)]^2 \\
%& = &
%4 \mu^2~ [\phi(\Delta_1)]^2~ [\Delta_1^2~ - 1] - 4 \mu^2 [\phi(\Delta_1)]^2~ \Delta_1^2 \\
& = &
- 4 \mu^2 [\phi(\Delta_1)]^2 ~<~ 0.
\end{eqnarray*}

As $\det[H(a_0, b_0)] < 0$ and the partial derivatives of the first order vanish at $(a_0,b_0)$ we conclude that $(a_0, b_0)$ is a saddle point.

%\subsection{Dev restricted to the valley}
%\label{app:spmed:valley}
%
%
%To  analyze $\frac{\partial \Delta_1}{\partial b}$ we write
%$\Delta_1=\frac{e_1}{\sigma_1}$ and $\Delta_2=\frac{e_2}{\sigma_2}$
%where $e_1=c_1-a\mu$ is the spare capacity in bin $1$ and
%$e_2=c_2-(1-a)\mu$ is the spare capacity in bin $2$. We notice that
%$e=e_1+e_2=c-\mu$ the total spare capacity in the system. Now
%$\Delta_1=\Delta_2$ implies $e_1\sigma_2=e_2
%\sigma_1=(e-e_1)\sigma_1$. Therefore, $e_1(\sigma_1+\sigma_2)=e
%\sigma_1$ and
%$\Delta_1=\frac{e}{\sigma_1+\sigma_2}=\frac{c-\mu}{\sigma}(\frac{1}{\sqrt{b}+\sqrt{1-b}})$
%and notice that $\Delta=\frac{c-\mu}{\sigma}$ is independent of $b$.
%All that remains is to differentiate the function
%$\frac{1}{\sqrt{b}+\sqrt{1-b}}$.
%%\begin{eqnarray*}
%%\frac{\partial \Delta_1}{\partial b}
%%& = &
%%- \frac{C-D}{2 \sigma}~ \frac{1}{(\sqrt{b} + \sqrt{1-b})^2}~ (\frac{1}{\sqrt{b}} - \frac{1}{\sqrt{1-b}})
%%\end{eqnarray*}
%
%One can also see that if $(a,b)$ is on the valley then $a=\frac{\sqrt{b}}{\sqrt{b}+\sqrt{1-b}} -\frac{\sqrt{b} c_2 - \sqrt{1-b} c_1}{(\sqrt{b}+\sqrt{1-b})\mu}$. In particular for $b=0$ we get the point $(\frac{c_1}{\mu},0)$ and for $b=1$ we get the point $(1-\frac{c_2}{\mu},1)$.
%%The $a$ derivative always vanishes on the valley and therefore also on these points. Also on the valley the deviation is $D(b)=(c-\mu)\frac{g(\Delta_1)}{\Delta_1}$ (see the proof of Lemma ?).

\section{SP-MWOP}
\subsection{Expected overflow probability of a single bin}

The overflow probability of bin $j$, denoted by $OFP_{S_j}$ is
$OFP_{S_j}(\mu_j,V_j)
 =
\frac{1}{\sigma_j \sqrt{2 \pi}} \int_{c_j}^{\infty}
e^{-\frac{(x-\mu_j)^2}{2\sigma_j^2}} dx$. Substituting $t = \frac{x
- \mu_j}{\sigma_j}$ we get

\begin{eqnarray*}
OFP_{S_j}(\mu_j,V_j) & = & \frac{1}{\sqrt{2 \pi}} \int_{\frac{c_j -
\mu_j}{\sigma_j}}^{\infty} e^{-\frac{t^2}{2}} dt   =  1 -
\Phi(\frac{c_j - \mu_j}{\sigma_j})  ~=~ 1 - \Phi(\Delta_j).
\end{eqnarray*}
\label{app:spmwop:single_bin}

\section{SP-MOP}
\label{sec:spmop}
\subsection{Proof of Lemma \ref{lem:OFP}}

\begin{proof}
The condition $\frac{\partial OFP}{\partial a}=0$ is equivalent to
\begin{eqnarray*}
\frac{\phi(\Delta_1)}{\Phi(\Delta_1)} & = &
\frac{\sigma_1}{\sigma_2} \cdot
\frac{\phi(\Delta_2)}{\Phi(\Delta_2)}
\end{eqnarray*}
As $b < \half$, $b < 1-b$ and $\sigma_1 < \sigma_2$. Hence,
\begin{eqnarray*}
\frac{\phi(\Delta_1)}{\Phi(\Delta_1)} & < &
\frac{\phi(\Delta_2)}{\Phi(\Delta_2)}.
\end{eqnarray*}
Denote $h(\Delta)=\frac{\phi(\Delta)}{\Phi(\Delta)}$. We will prove
that $h$ is monotone decreasing, and this implies that $\Delta_1 >
\Delta_2$.

To see that $h$ is monotone decreasing define $H(\Delta)=\ln
(\Phi(\Delta))$. Then $h=H'$. Therefore, $h'=H''$. However, $\Phi$
is log-concave, hence $H'' <0$. We conclude that $h' < 0$ and $h$ is
monotone decreasing.
\end{proof}

\section{Unbalancing bin capacities is always better}
\label{app:unbalance_bin_capacities}

Suppose we are given a capacity budget $c$ and we have the freedom to choose capacities $c_1,c_2$ that sum up to $c$ for two bins. Which choice is the best? Offhand, it is possible that for each input there is a different choice of $c_1$ and $c_2$ that minimizes the expected deviation. In contrast, we show that the minimum expected deviation always decreases as the difference $c_2 - c_1$ increases.

\begin{lemma}
\label{lem:ubc} % UnBallance Capacity
Given a capacity budget $c$, the minimum expected deviation decreases as $c_2 - c_1$ increases. In particular the best choice is having a single bin with capacity $c$ and the worst choice is splitting the capacities evenly between the two bins.
\end{lemma}

\begin{proof}
Recall that $\Delta_1(a,b) = \frac{c_1 - a\mu}{\sigma \sqrt{b}}$ and $\Delta_2(a,b) = \frac{c_2 - (1-a)\mu}{\sigma \sqrt{1-b}}$. Therefore, if we reduce $c_1$ by $\tilde{c}$ and increase $c_2$ by $\tilde{c}$, we get
$$\tilde{\Delta}_1(a,b) = \frac{c_1 - \tilde{c} - a\mu}{\sigma \sqrt{b}} = \frac{c_1 - (a - \frac{\tilde{c}}{\mu})\mu}{\sigma \sqrt{b}} = \Delta_1(a-\frac{\tilde{c}}{\mu},b).$$

Similarly, $\tilde{\Delta}_2(a,b) = \Delta_2(a-\frac{\tilde{c}}{\mu},b)$. Let $Dev_{c_1,c_2}(a,b)$ denote the expected deviation with bin capacities $c_1,c_2$.
As $Dev(a,b)=\sigma_1(b) g(\Delta_1(a,b))+\sigma_2(b) g(\Delta_2(a,b))$ we see that
$$Dev_{c_1-\tilde{c}, c_2+\tilde{c}}(a,b)=Dev_{c_1,c_2}(a-\frac{\tilde{c}}{\mu},b),$$
i.e., the graph is shifted  left by $\frac{\tilde{c}}{\mu}$.

Notice that the bottom sorted path does not depend on the bin capacities and is the same in both cases. Let $(a,b)$ be the optimal fractional solution for bin capacities $c_1,c_2$. We know that $(a,b)$ is on the bottom sorted path. Let $\tilde{a}=a-\frac{\tilde{c}}{\mu}$. We know that $Dev_{c_1-\tilde{c}, c_2+\tilde{c}}(\tilde{a},b)=Dev_{c_1,c_2}(a,b)$. The point $(\tilde{a},b)$ lies to the left of the bottom sorted path and therefore above it. As the optimal solution for bin capacities $c_1-\tilde{c}, c_2+\tilde{c}$ is also on the bottom sorted path and is strictly better than any internal point, we conclude that the expected deviation for bin capacities $c_1-\tilde{c}, c_2+\tilde{c}$ is strictly smaller than the expected deviation for bin capacities $c_1,c_2$.
\end{proof}

An immediate corollary is the trivial fact that putting all the capacity budget in one bin is best. Obviously, this is not always possible nor desirable, but if there is tolerance in each bin capacity, we recommend minimizing the number of bins.

\end{document}